\documentclass{fundam}

\usepackage[T1]{fontenc}
\usepackage[utf8]{inputenc}
 \usepackage{amssymb}
\usepackage{url}
\usepackage{hyphenat}
\usepackage{xcolor}
\usepackage[all]{xy}
\usepackage{amsfonts}
\usepackage{mathrsfs}
\usepackage{graphicx}

\newcommand{\labell}[1]{\POS[]\POS!L\drop{\strut\llap{$#1$}\hspace{2ex}}}
\newcommand{\labelu}[1]{\POS[]\POS!U\drop{\raisebox{1.2em}{\strut$#1$}}}
\newcommand{\labeldr}[1]{\POS[]\POS!D(2)!R(1.5)\drop{\strut\rlap{$#1$}\hspace{2ex}}}
\newcommand{\labelr}[1]{\POS[]\POS!R\drop{\strut\hspace{2ex}\rlap{$#1$}}}
\newcommand{\labeld}[1]{\POS[]\POS!D(3)\drop{\strut\rlap{$#1$}\hspace{2ex}}}
\newcommand{\verticale}{\ar@{--}[d]}

\newsavebox{\labelz}\newsavebox{\labelun}
\savebox{\labelz}{$\xymatrix@M=0.2ex{*++[F]{0}}$}
\savebox{\labelun}{$\xymatrix@M=0.2ex{*++[F]{1}}$}

\def\rebar#1{\expandafter\def\csname #1bar\endcsname{\overline{\csname
      #1\endcsname}}}		%

\def\M{\mathcal{M}}

\newcommand{\N}{{\mathcal{N}}}
\newcommand{\X}{{\mathcal{X}}}
\newcommand{\FFF}{\mathfrak{F}}
\def\bbR{\mathbb{R}}
\newcommand{\bbZ}{\mathbb{Z}}
\renewcommand{\P}{\mathcal{P}}
\newcommand{\R}{\mathscr{R}}

\rebar{M}	%
\rebar{N}	%

\newcommand{\C}{\mathscr{C}}

\newcommand{\Dstar}{\mathfrak{D}}

\newcommand{\DCS}{{\text{\normalfont\sffamily DCS}}}

\newcommand{\up}[1]{\,\uparrow #1}

\newcommand{\Cstar}{\mathfrak{C}}

\newcommand{\slgb}{\mbox{$\sigma$-al}\-ge\-bra}

\newcommand{\BM}{\partial\M}

\newcommand{\lub}{\textsl{l.u.b.}}
\newcommand{\glb}{\textsl{g.l.b.}}

\newcommand{\vd}{\varepsilon}

\newcommand{\tq}{\;\big|\;}
\newcommand{\tqs}{\;:\;}

\newcommand{\ie}{\textsl{i.e.}}

\newcommand{\rest}[1]{\bigl|_{#1}}

\newcommand{\un}{\mathbf{1}}

\newcommand{\HPOV}{\emph{Heaps point of view}}
\renewcommand{\H}[1]{\widehat{#1}}

\numberwithin{equation}{section}

\newsavebox{\framedzero}\savebox{\framedzero}{\raisebox{1ex}{\xymatrix{*+[F]{0}}}}

\newlength{\tempa}

\newlength{\tempcc}
\newlength{\hauttrans}
\newlength{\longtrans}
\newlength{\longplace}

\newcommand{\petriexec}{%
\setlength{\hauttrans}{.66ex}
\setlength{\longtrans}{6ex}
\setlength{\longplace}{4ex}
\settoheight{\tempcc}{\strut}
\settodepth{\tempa}{\strut}
\addtolength{\tempcc}{\tempa}
\setlength{\tempcc}{2\tempcc}
\addtolength{\tempcc}{\hauttrans}}

\newcommand{\transh}{*+[F]\txt{\vbox to\hauttrans{\hbox to\longtrans{}}}}
\newcommand{\transhp}{*+[F--]\txt{\vbox to\hauttrans{\hbox to\longtrans{}}}}
\newcommand{\transv}{*+[F]\txt{\vbox to\longtrans{\hbox to\hauttrans{}}}}
\newcommand{\transvp}{*+[F--]\txt{\vbox to\longtrans{\hbox to\hauttrans{}}}}
\newcommand{\place}{*+[o][F]{\vbox to\longplace{\hbox to\longplace{}}}}

\newcommand{\transht}{*+[F.]\txt{\vbox to\hauttrans{\hbox to\longtrans{}}}}
\newcommand{\placet}{*+[o][F.]{\vbox to\longplace{\hbox to\longplace{}}}}

\newcommand{\markp}{\POS[]\drop{\bullet}}

\petriexec

\publyear{22}
\papernumber{2133}
\volume{187}
\issue{2-4}

 \finalVersionForARXIV

\begin{document}

\title{Introduction to Probabilistic Concurrent Systems}

\author{Samy Abbes\thanks{Address for correspondence:  Laboratoire IRIF, Université Paris Cité,
                                         Case 7014, 75205 Paris Cedex 13, France}
 \\
 Laboratoire IRIF, Université Paris Cité\\
Paris, France\\
abbes@irif.fr
}

 \runninghead{S. Abbes}{Introduction to Probabilistic Concurrent Systems}

\maketitle

\begin{abstract}
The first part of the paper is an introduction to the theory of probabilistic concurrent systems under a partial order semantics. Key definitions and results are given and illustrated on examples.

The second part includes contributions. We introduce deterministic concurrent systems as a subclass of concurrent systems. Deterministic concurrent system are ``locally commutative'' concurrent systems. We prove that irreducible and deterministic concurrent systems have a unique probabilistic dynamics, and we characterize these systems by means of their combinatorial properties.

\noindent ACM CSS: G.2.1; F.1.1
\end{abstract}

\section{Introduction}
\label{sec:introduction}

Trace monoids are well known models of concurrency, typically used when one wishes to work on the logical order between actions rather than on their chronological order. These models represent systems with actions, symbolized by letters in a given alphabet, and with the feature that some actions may occur concurrently. Let $a_1,\ldots,a_N$ be a bunch of pairwise concurrent actions about to be played  during an execution of the system. Then the system does not distinguish between the $N!$ possible ways of interleaving them; nor could an observer retrieve any information on their interleaving. When observing the system history, the only remaining information about these $N$ actions is that they were performed concurrently; and actually it would be irrelevant to think of a ``hidden interleaving''.

Mathematically, a trace monoid $\M$ is a monoid generated by an alphabet~$\Sigma$, and with relations of the form $ab=ba$ for some fixed pairs of letters $(a,b)\in\Sigma\times\Sigma$. The identity $ab=ba$ in $\M$ renders the concurrency of the two actions $a$ and~$b$. This identity is typical of the so-called \emph{partial order} or \emph{true-concurrent} semantics for concurrency. It contrasts with the \emph{interleaving semantics}, which would instead keep track of the two possible sequences $a$-then-$b$ and $b$-then-$a$ when facing the two concurrent actions $a$ and~$b$.

Despite their successful use as models of concurrency for databases for instance~\cite{diekert90,diekert95}, trace monoids lack an essential feature present in most real-life systems, namely they lack a notion of state. Indeed, any action can be performed at any time when considering a trace monoid model; whereas, in real-life systems, some actions may only be enabled when the system enters some specified state, and then one expects the system to enter a new state, determined by the former state and by the action performed.

A natural model combining both the ``built-in'' concurrency feature of trace monoids and the notion of state arises when considering a right monoid action of a trace monoid $\M$ on a finite set of states~$X$, \ie, a mapping $X\times \M\to X$ denoted by $(\alpha,x)\mapsto\alpha\cdot x$. A sink state $\bot$ is introduced in order to distinguish the forbidden actions. Hence, if the system is in state~$\alpha$, performing the letter $a\in\Sigma$ brings the system into the new state~$\alpha\cdot a$, with the convention that $a$ was actually not allowed if $\alpha\cdot a=\bot$. This notion of concurrent system, introduced in~\cite{abbes19:_markov}, encompasses in particular popular models of concurrency such as bounded Petri nets~\cite{reisig85,nielsen81}.

\medskip
Whereas the interleaving semantics of systems provides a direct connection with the classical theory of probabilistic systems (Markov chains in continuous or in discrete time, mainly), adding a probabilistic layer on top of concurrency models within the partial order semantics has been a challenge for some time. Indeed, there is no obvious way to assign a probability to traces with a ``natural'' composition property. The random walk approach for instance, consisting in adding one letter (or action) at a time with each letter being assigned a fixed probability, can be shown to never fulfill the composition property that we are looking for (see a more detailed discussion in Remark~\ref{rem:2} in Sect.~\ref{sec:prob-dynam-trace}, \S~\emph{Probabilistic dynamics for trace monoids}).

The approach of the author and his co-authors on this topic has been to start again from the very beginning: trace monoids themselves~\cite{abbes15}. Equipping the ``trajectories'' of a trace monoid with a natural probabilistic dynamics amounts to defining a memoryless probability measure on the space of ``infinite traces''. The memoryless property of measures is a natural composition property which extends, in the framework of trace monoids, the well known memoryless property typical of, say, coin tossing. The existence of such measures for trace monoids is not obvious. The construction of memoryless probability measures for trace monoids is based on existing tools found in the literature on their combinatorics originally due to  Cartier and Foata~\cite{cartier69} and later revisited by Viennot~\cite{viennot86}. The construction puts into light, in the elementary framework of trace monoids, some essential concepts for the interplay between probability and concurrency: the Möbius polynomial of the monoid and the particular role played by its root of smallest modulus, and the process of cliques visited by an infinite random trace. For the more complex case of concurrent systems, defining a probabilistic dynamics consists then in a more technical work on the very same concepts. New difficulties arise in this case, yet a general theory of probabilistic concurrent systems may be built~\cite{abbes19:_markov,abbes20}.

\eject
The purpose of the present paper is twofold. Firstly, it intends to present an introduction to the theory of probabilistic concurrent systems. We present the key notions and state the main results which justify and guide the computations to be done. Results are stated in a rigorous way, but we do not provide proofs (references are given). The hope is to provide an elementary introduction both to trace theory, from the systems theory point of view, and to the probabilistic aspects of concurrent systems. This includes basic definitions of trace monoids, the Cartier-Foata normal form for traces, the Möbius polynomial and the Möbius transform associated with a trace monoid and the notion of irreducibility for trace monoids. A rigorous notion of infinite trace is developed, and a characterization of memoryless probability measures for trace monoids is given. The particular case of the uniform measure is investigated. The realization result of memoryless measures as finite Markov chains on cliques is also introduced. All these notions are then developed in the more general context of concurrent systems, which yields us to introduce the following notions attached to a concurrent system: its Möbius matrix, its characteristic root, its digraph of states-and-cliques, its Markov measures, and among them its uniform measure in the irreducible case. A key result of the theory is the spectral property for irreducible concurrent systems, to be used later in the paper. The computation of the probability distribution of the first clique of a random infinite execution is illustrated on several examples, and the notion of null node is introduced.

Secondly, and on the contribution part, we introduce and investigate the special case of  \emph{deterministic concurrent systems}. Intuitively, a deterministic concurrent system (\DCS) is a concurrent system where no conflict between different actions can ever arise. Deterministic concurrent systems can be related, for instance, to causal nets and to elementary event structures found in 1980's papers~\cite{nielsen81}. We prove in particular that deterministic concurrent systems correspond to concurrent systems which are ``locally commutative''.

Compared to general concurrent systems, deterministic concurrent systems appear as limit cases. For instance, we prove that their space of infinite executions is at most countable---whereas it is uncountable in general. If the system is moreover irreducible, we show that from any initial state, only one infinite execution exists. In particular the only probabilistic dynamics is trivial in this case---whereas there is a continuum of possible and non trivial probabilistic dynamics in general. Yet, proving these properties is not obvious. The definition of \DCS\ is formulated in elementary terms; their specific properties are formulated in elementary terms; but the proof of these properties relies on some subtle combinatorics of partially ordered sets.

We state general properties of deterministic concurrent systems, and our main contribution is to give several equivalent characterizations of concurrent systems which are both deterministic and irreducible: an algebraic characterization; a probabilistic characterization; a characterization from the analytic combinatorics viewpoint; and a characterization through set-theoretic properties of the set of infinite executions.

Another contribution is a generalization of the well known fact that commutative free monoids have a polynomial growth. The property that we obtain in Corollary~\ref{cor:1} is general enough to be of interest \emph{per se}.\vspace*{-3mm}

\paragraph*{Organization of the paper.}
\label{sec:organisation-paper}

Section~\ref{sec:preliminaries} is devoted to the background on concurrent probabilistic systems, and is divided into three parts. Sections~\ref{sec:trace-monoids-their} surveys basic notions on trace monoids and their probabilistic counterpart, while Section~\ref{sec:conc-syst-their} reviews basic constructions related to concurrent systems, including the probabilistic notions and their relationship with combinatorics. Finally, Section~\ref{sec:an-elem-comp} is devoted to an elementary, yet original result of trace theory, that will be used later in the paper. Deterministic concurrent systems are introduced in Section~\ref{sec:determ-conc-syst}. Section~\ref{sec:irred-determ-conc} is devoted to the study of concurrent systems which are both deterministic and irreducible.

\section{Preliminaries}
\label{sec:preliminaries}

\subsection{Trace monoids and their combinatorics}
\label{sec:trace-monoids-their}

The background material introduced in this section is standard, see for instance~\cite{diekert90,diekert95}, excepted for the probabilistic notions which are borrowed from~\cite{abbes15}.\vspace*{-2mm}

\paragraph*{Independence and dependence pairs.}
An \emph{alphabet} is a finite set, which we usually denote by~$\Sigma$, the elements of which are called \emph{letters}. An  \emph{independence pair} is a pair $(\Sigma,I)$, where $I$ is a binary symmetric and irreflexive relation on~$\Sigma$, called an \emph{independence relation}. A \emph{dependence pair} is a pair $(\Sigma,D)$, where $D$ is a binary symmetric and reflexive relation on~$\Sigma$, called a \emph{dependence relation}. With $\Sigma$ fixed, dependence and independence relations correspond bijectively to one another, through the association $D=(\Sigma\times\Sigma)\setminus I$. The \emph{Coxeter graph} of either pair $(\Sigma,I)$ or  $(\Sigma,D)$ is the graph $(\Sigma,D)$ with all self-loops omitted~\cite{dehornoy15}.

\begin{example}
  \label{exm:1}
    Figure~\ref{fig:coacosn} depicts the Coxeter graph of the independence pair $(\Sigma,I)$ with $\Sigma=\!\{a_0,...\,,a_4\}$ and $(a_i,a_j)\in I\iff |i-j|\geq2$.
\end{example}

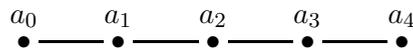
\begin{figure}[!h]
\vspace*{-7mm}
  \centering
  $$
\xymatrix{\bullet\ar@{-}[r]\labelu{a_0}&\bullet\ar@{-}[r]\labelu{a_1}&\bullet\ar@{-}[r]\labelu{a_2}&\bullet\ar@{-}[r]\labelu{a_3}&
\bullet\labelu{a_4}}\vspace*{-5mm}
  $$
  \caption{\small Coxeter graph of the independence pair $(\Sigma,I)$ with $\Sigma=\{a_0,\ldots,a_4\}$ and $(a_i,a_j)\in I\iff |i-j|\geq2$. The arcs of the Coxeter graph correspond to the pairs $(a_i,a_j)$ with $|i-j|=1$}
  \label{fig:coacosn}
\end{figure}

With the alphabet $\Sigma$ fixed, independence pairs are ordered by inclusion and form a sub-lattice of~$\P(\Sigma\times\Sigma)$. The minimum is $I_0=\emptyset$ and the maximum is $I_1=(\Sigma\times\Sigma)\setminus\Delta$, where $\Delta$ is the diagonal relation $\Delta=\{(x,x)\tq x\in\Sigma\}$.\vspace*{-2mm}

\paragraph*{Traces and trace monoids.}
The \emph{trace monoid} $\M=\M(\Sigma,I)$ is the monoid with generators and relations with the following presentation:
\begin{gather*}
  \M=\langle\Sigma\;\big|\; \forall (a,b)\in I\quad ab=ba\rangle.
\end{gather*}

By definition (see, for instance, \cite[Chap.~7]{book93}~for presentations of monoids), $\M$~is the quotient monoid $\Sigma^*/\R$, where $\R$ is the congruence on $\Sigma^*$ generated by all pairs $(ab,ba)$ for $(a,b)$ ranging over~$I$. Elements of $\M$ are called \emph{traces}. Hence every trace is the congruence class of some word of~$\Sigma^*$; and two words of $\Sigma^*$ are congruent whenever they can be obtained from one another by applying arbitrary many times the following rewriting rule:
\begin{gather*}
  \text{for $(a,b)\in I$ and $x,y\in\Sigma^*$:}\qquad xaby\longrightarrow xbay
\end{gather*}

The trace monoid $\M$ is \emph{non trivial} if $\Sigma\neq\emptyset$, in which case $\M$ is countably infinite.

We denote by $\pi_I:\Sigma^*\to\M$ the canonical morphism.  The unit element of~$\M$, image of the empty word, is called the \emph{empty trace} and is denoted by~$\vd$. The concatenation of $x,y\in\M$ is denoted by~$x\cdot y$. We identify letters of the alphabet $\Sigma$ with their images in~$\M$ through the canonical mappings $\Sigma\to\Sigma^*\xrightarrow{\pi_I}\M$. By construction, any two distinct letters $a$ and $b$ commute in $\M$ if and only if $(a,b)\in I$.

Two extreme cases of trace monoids correspond to the extremal independence relations introduced above~: $\M(\Sigma,I_0)$ is isomorphic to the free monoid~$\Sigma^*$, where no two distinct letters commute with each other; $\M(\Sigma,I_1)$~is isomorphic to the free commutative monoid, where all letters commute with each other. In the general case, only some pairs of distinct letters commute with each other, namely those not directly connected in the Coxeter graph; hence the alternative name of \emph{free partially commutative monoids} for trace monoids in the literature.\vspace*{-2mm}

\paragraph*{Representation of traces.}
\label{sec:repr-trac}

\emph{Heaps of pieces} are combinatorial objects introduced in~\cite{viennot86} which provide an intuitive  visual representation of traces (see also~\cite{krattenthaler06}). Picture each letter as a piece falling to the ground, in such a way that distinct letters which commute with each other fall along parallel and disjoint lines; whereas non commutative letters fall in such a way that they block each other. The heaps of pieces thus obtained are combinatorial objects corresponding bijectively to the elements of the trace monoid, by reading the letters labeling the pieces from bottom to top. Several readings are possible, corresponding to the different words in the congruence class of the trace.

\begin{example}
Figure~\ref{fig:apoaspokzx} (left) depicts the heap of pieces corresponding to a trace of the monoid $\M(\Sigma,I)$ from Example~\ref{exm:1}.
\end{example}

\begin{figure}[!h]
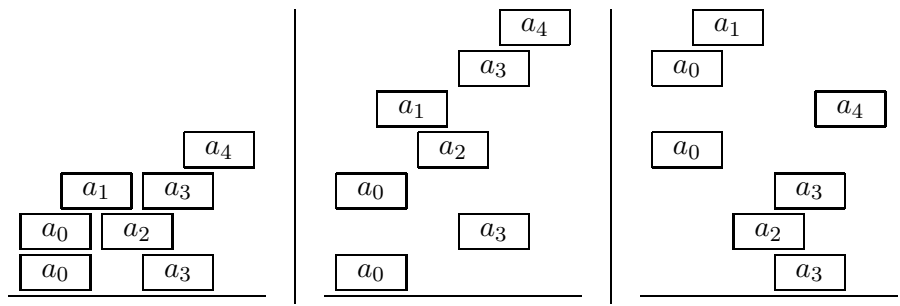

  \centering
    \begin{tabular}{c@{\quad}|@{\quad}c@{\quad}|@{\quad}c}
\xy
<.1em,0em>:
0="G",
"G"+(12,6)*{a_0},
"G";"G"+(24,0)**@{-};"G"+(24,12)**@{-};"G"+(0,12)**@{-};"G"**@{-},
(42,0)="G",
         "G"+(12,6)*{a_3},
"G";"G"+(24,0)**@{-};"G"+(24,12)**@{-};"G"+(0,12)**@{-};"G"**@{-},
(0,14)="G";
         "G"+(12,6)*{a_0},
"G";"G"+(24,0)**@{-};"G"+(24,12)**@{-};"G"+(0,12)**@{-};"G"**@{-},
(28,14)="G";
         "G"+(12,6)*{a_2},
"G";"G"+(24,0)**@{-};"G"+(24,12)**@{-};"G"+(0,12)**@{-};"G"**@{-},
(14,28)="G";
         "G"+(12,6)*{a_1},
"G";"G"+(24,0)**@{-};"G"+(24,12)**@{-};"G"+(0,12)**@{-};"G"**@{-},
(42,28)="G";
         "G"+(12,6)*{a_3},
"G";"G"+(24,0)**@{-};"G"+(24,12)**@{-};"G"+(0,12)**@{-};"G"**@{-},
(56,42)="G";
         "G"+(12,6)*{a_4},
"G";"G"+(24,0)**@{-};"G"+(24,12)**@{-};"G"+(0,12)**@{-};"G"**@{-},
(-4,-2);(84,-2)**@{-}
\endxy
&
 \xy
<.1em,0em>:
0="G",
"G"+(12,6)*{a_0},
"G";"G"+(24,0)**@{-};"G"+(24,12)**@{-};"G"+(0,12)**@{-};"G"**@{-},
(42,14)="G",
         "G"+(12,6)*{a_3},
"G";"G"+(24,0)**@{-};"G"+(24,12)**@{-};"G"+(0,12)**@{-};"G"**@{-},
(0,28)="G";
         "G"+(12,6)*{a_0},
"G";"G"+(24,0)**@{-};"G"+(24,12)**@{-};"G"+(0,12)**@{-};"G"**@{-},
(28,42)="G";
         "G"+(12,6)*{a_2},
"G";"G"+(24,0)**@{-};"G"+(24,12)**@{-};"G"+(0,12)**@{-};"G"**@{-},
(14,56)="G";
         "G"+(12,6)*{a_1},
"G";"G"+(24,0)**@{-};"G"+(24,12)**@{-};"G"+(0,12)**@{-};"G"**@{-},
(42,70)="G";
         "G"+(12,6)*{a_3},
"G";"G"+(24,0)**@{-};"G"+(24,12)**@{-};"G"+(0,12)**@{-};"G"**@{-},
(56,84)="G";
         "G"+(12,6)*{a_4},
"G";"G"+(24,0)**@{-};"G"+(24,12)**@{-};"G"+(0,12)**@{-};"G"**@{-},
(-4,-2);(84,-2)**@{-}
\endxy
&
 \xy
<.1em,0em>:
(0,70)="G",
"G"+(12,6)*{a_0},
"G";"G"+(24,0)**@{-};"G"+(24,12)**@{-};"G"+(0,12)**@{-};"G"**@{-},
(42,0)="G",
         "G"+(12,6)*{a_3},
"G";"G"+(24,0)**@{-};"G"+(24,12)**@{-};"G"+(0,12)**@{-};"G"**@{-},
(0,42)="G";
         "G"+(12,6)*{a_0},
"G";"G"+(24,0)**@{-};"G"+(24,12)**@{-};"G"+(0,12)**@{-};"G"**@{-},
(28,14)="G";
         "G"+(12,6)*{a_2},
"G";"G"+(24,0)**@{-};"G"+(24,12)**@{-};"G"+(0,12)**@{-};"G"**@{-},
(14,84)="G";
         "G"+(12,6)*{a_1},
"G";"G"+(24,0)**@{-};"G"+(24,12)**@{-};"G"+(0,12)**@{-};"G"**@{-},
(42,28)="G";
         "G"+(12,6)*{a_3},
"G";"G"+(24,0)**@{-};"G"+(24,12)**@{-};"G"+(0,12)**@{-};"G"**@{-},
(56,56)="G";
         "G"+(12,6)*{a_4},
"G";"G"+(24,0)**@{-};"G"+(24,12)**@{-};"G"+(0,12)**@{-};"G"**@{-},
(-4,-2);(84,-2)**@{-}
\endxy
    \end{tabular}
    \caption{\small In this example the commutation relations are those of Example~\ref{exm:1}, with Coxeter graph depicted on Fig.~\ref{fig:coacosn}. \textsl{Left:} representation as a heap of piece of the trace $y=a_0\cdot a_3\cdot a_0\cdot a_2\cdot a_1\cdot a_3\cdot a_4$.
\textsl{Middle and right:} representations of two words in the congruence class of~$y$: $a_0$-$a_3$-$a_0$-$a_2$-$a_1$-$a_3$-$a_4$ (middle) and $a_3$-$a_2$-$a_3$-$a_0$-$a_4$-$a_0$-$a_1$ (right)}
  \label{fig:apoaspokzx}
\end{figure}

If $x$ is a trace, we denote by $\H x$ the corresponding heap. The identification of traces with heaps is sound in the following sense. If $x$ and $y$ are two traces, then the heap $\H{x\cdot y}$ is obtained by piling up the two heaps $\H x$ and~$\H y$, and then letting pieces from $\H y$ fall down to the ground or until they are blocked by pieces from~$\H x$, which produces a new arrangement of the resulting heap.\vspace*{-2mm}

\paragraph*{Length, occurrence of letters and divisibility order.}

By its very construction as a quotient monoid, $\M$~comes equipped with a number of objects for which we give some details now. Let $x\in\M$ be the congruence class of a word $u\in\Sigma^*$. The \emph{length} of~$x$, denoted by~$|x|$, is the length of~$u$. The quantity $|u|$ is independent of the choice of~$u$. The length is additive~: $|x\cdot y|=|x|+|y|$, and satisfies $|x|=0\iff x=\vd$. \HPOV: $|x|$~represents the number of pieces in the heap~$\H x$.

Furthermore, for each letter $a\in\Sigma$, we write $a\in x$ whenever $a$ has at least one occurrence in~$u$, and we write $a\notin x$ otherwise. \HPOV: $a\in x$ means that the heap $\H x$ contains a piece labeled by~$a$.

Finally, the preorder $(\M,\leq)$ inherited from the left divisibility in $\M$ is defined by: $x\leq y\iff(\exists z\in\M\quad y=x\cdot z)$. This preorder is actually a partial order since $\M$ is equipped with the length function introduced above (the antisymmetry of $\leq$ derives at once from the existence of the length function). \HPOV: $x\leq y$ whenever one can complete the heap $\H x$ by letting additional letters fall from the top and obtain the heap~$\H y$. In this case we say that  $\H x$ is a \emph{sub-heap} of~$\H y$.

The monoid $\M$ is \emph{left cancellative}: for $x,y,z\in\M$, if $x\cdot y=x\cdot z$ then $y=z$ (the proof given in~\cite{cartier69} is based on the existence of a normal form for traces, see below).

As a consequence, if $x,y\in\M$ are such that $x\leq y$, the element $z\in\M$ such that $y=x\cdot z$ is unique. We denote this element by $z=x\backslash y$. \HPOV: the heap $\H z$ is obtained by removing from below in the heap $\H y$ the pieces that form the heap~$\H x$, sub-heap of~$\H y$.\vspace*{-2mm}

\paragraph*{Cliques.}
A \emph{clique} of $\M$ is a trace of the form $x=a_1\cdot\ldots\cdot a_k$, where all $a_i$s are letters such that $i\neq j\implies (a_i,a_j)\in I$. Hence a clique represents a set of mutually concurrent actions.

We denote by $\C$ the set of cliques, which is a finite set. Letters and the empty trace are cliques of length $1$ and $0$ respectively. There exist cliques of length at least $2$ if and only if $\M$ is not a free monoid, or equivalently, if the independence relation is not empty. \HPOV:  heaps corresponding to cliques are the horizontal ones, with all pieces directly on the ground.

Since all $a_i$s commute with each other, we identify the clique $x=a_1\cdot\ldots\cdot a_k$ with the subset $\{a_1,\ldots,a_k\}\in\P(\Sigma)$. Through this identification,  $(\C,\leq)$~is isomorphic to a downward-closed subset of \mbox{$(\P(\Sigma),\subseteq)$}. It corresponds to the full powerset $(\P(\Sigma),\subseteq)$ if and only if $\M$ is the free commutative monoid on~$\Sigma$.

A \emph{non empty clique} is a clique $x\neq\vd$. The set of non empty cliques of~$\M$ is denoted by~$\Cstar$. Minimal elements of $(\Cstar,\leq)$ correspond to the letters of~$\Sigma$.

\begin{example}
  For the monoid from Example~\ref{exm:1}, the set of cliques is the following:
\begin{gather*}
  \begin{array}{l@{\qquad}l}
    \C=\{\;\vd&\text{$1$ clique of length $0$}\\[-1pt]
    \ \ a_0,\,a_1,\,a_2,\,a_3,\,a_4,&\text{$5$ cliques of length $1$}\\[-1pt]
     \ \     a_0\cdot a_2,\, a_0\cdot a_3,\, a_0\cdot a_4,\ a_1\cdot a_3,\,a_1\cdot a_4,\,a_2\cdot a_4,&
                                                                                                      \text{$6$ cliques of length $2$}\\[-1pt]
     \ \   a_0\cdot a_2\cdot a_4\;\}&\text{$1$ clique of length $3$}
  \end{array}
\end{gather*}\vspace*{-4mm}
\end{example}

\paragraph*{Lower and upper bounds of traces.}
\label{sec:lower-upper-bounds}

Any two traces $x,y\in\M$ have a greatest lower bound (\glb) in $(\M,\leq)$, which we denote by $x\wedge y$. \HPOV: the heap corresponding to $x\wedge y$ is obtained as the maximal common sub-heap of $\H x$ and of~$\H y$. In the case where both $x$ and $y$ are cliques, then $x\wedge y$~is the clique corresponding to the subset $x\cap y\in\P(\Sigma)$.
\vfil\eject

Two traces $x$ and $y$ have a least upper bound (\lub) in $(\M,\leq)$, denoted by $x\vee y$ if it exists, if and only if they have a common upper
bound.\vspace*{-2mm}

\paragraph*{Normal sequences. Normal form and generalized normal form of traces.}
\label{sec:normal-sequences}

Cartier and Foata have introduced in \cite{cartier69} a normal form for traces\footnote{See \cite{epstein92} for a general notion of normal form in a quotient monoid. See~\cite{dehornoy15} for the description of a normal form for a class of presented monoids including trace monoids.}, which we describe now.

\medskip
Let $\M=\M(\Sigma,I)$ be a trace monoid with associated dependence relation~$D$. A pair $(x,y)\in\C\times \C$ is a \emph{normal pair}, which we denote by $x\to y$, if:
\begin{gather}
  \label{eq:11}
\forall b\in y\quad\exists a\in x\quad(a,b)\in D.
\end{gather}

\HPOV: the pair $(x,y)\in\C\times\C$ is normal if and only if the horizontal heap $\H x$ can \emph{support} the horizontal heap~$\H y$, in the sense that $\H y$ can be piled up upon $\H x$ without any of its pieces falling down.

In any trace monoid, two particular cases occur: $x\to\vd$ for all $x\in\C$, and $\vd\to x$ if and only if $x=\vd$.

\begin{example}
Consider the clique  $x=a_0\cdot a_2$ in the trace monoid from Example~\ref{exm:1}. The non empty cliques  $y\in\Cstar$ such that $x\to y$ are the following: $a_0$, $a_1$, $a_2$, $a_3$, $a_0\cdot a_2$, $a_0\cdot a_3$, $a_1\cdot a_3$.
\end{example}

A sequence $(c_i)_i$ of cliques, the sequence being either finite or infinite, is a \emph{normal sequence} if $(c_i,c_{i+1})$ is a normal pair for all pairs of indices $(i,i+1)$.

The interest of this notion lies in the following result~\cite{cartier69}: {\itshape for any non empty trace~$x$, there exists a unique integer $k\geq1$ and a unique normal sequence $(c_1,\ldots,c_k)$ of non empty cliques such that $x=c_1\cdot\ldots\cdot c_k$}. The sequence $(c_1,\ldots,c_k)$ is the \emph{Cartier-Foata normal form of~$x$}, or the \emph{normal form of $x$} for short. The integer $k$ is the \emph{height} of~$x$.

\HPOV: the cliques that appear in the normal form of a trace $x$ correspond to the horizontal layers one sees in the heap~$\H x$. The height $k$ is the number of horizontal layers of~$\H x$.

\begin{example}[and warning]
  Let $y$ be the trace depicted on Fig.~\ref{fig:apoaspokzx}. It has the following normal form:  $(a_0\cdot a_3,\,a_0\cdot a_2,\,a_1\cdot a_3,\,a_4)$. Its height is~$4$.

\medskip
  Observe that if $x\leq y$, it does not imply that the normal form of $x$ is a prefix word of~$y$'s. Indeed, consider for instance $x=a_0a_0$. Then $x\leq y$ since, by the commutation relations in~$\M$, one has $y=xz$ with $z=a_3a_2a_1a_3a_4$. Yet, the normal form of $x$ is $(a_0,a_0)$, which is not a prefix word of the normal form of~$y$.

  Put differently, if $x$ is a trace with normal form $(c_1,\ldots,c_k)$, adding a letter $a$ or more generally a trace $z$ to $x$ yields a new trace $y=xz$ whose normal form $(d_1,\ldots,d_{k'})$ is not easily described from the normal forms of $x$ and of~$z$. In particular, the initial clique $d_1$ of $y$ may differ from~$c_1$, even if $z=a$ is a single letter since this letter might ``fall'' all the way down to the ground.
\end{example}

Since the height of traces varies, it is convenient to complete the normal form of traces as follows. For $x$ a non empty trace of height $k$ and with normal form $(c_1,\ldots,c_k)$, we put $c_i=\vd$ for all $i>k$. The now infinite sequence $(c_i)_{i\geq1}$ is still a normal sequence, called the \emph{generalized normal form of~$x$}. By convention, the generalized normal form of $\vd$ is the normal sequence $(\vd,\vd,\ldots)$.

\medskip
We observed above that the divisibility relation in a trace monoid $\M$ does not correspond to the prefix order on normal forms. More precisely, if $(c_i)_{i\geq1}$ is the generalized normal form of some trace $x\in\M$, and if $(d_i)_{i\geq1}$ is the generalized normal form of some trace $y\in\M$, then:
\begin{gather}
  \label{eq:19}
  x\leq y\text{ in $\M$}\iff\bigl(\forall i=1,2\ldots\quad c_i\leq d_i\text{ in $\C$}\,\bigr)
\end{gather}


\paragraph*{Generalized traces and infinite traces.}
\label{sec:generalised-traces}


The generalized normal forms of traces are constructed as infinite normal sequences of cliques. Conversely, let $\xi=(c_i)_{i\geq1}$ be an arbitrary infinite normal sequence of cliques. Then two cases may occur:
\begin{enumerate}
\item If $c_i=\vd$ for some integer $i\geq1$, then $c_j=\vd$ for all integers $j\geq i$. In this case, $\xi$~is the generalized normal form of some trace, namely of the trace $x=c_1\cdot\ldots\cdot c_i$.
\item Otherwise, $c_i\neq\vd$ for all integers $i\geq1$. Based on the heap of pieces intuition, it is natural to define such objects as infinite traces, since they correspond to infinite piles of layers.
\end{enumerate}

This motivates the following definitions. A \emph{generalized trace} is any infinite normal sequence $\xi=(c_i)_{i\geq1}$ of cliques. We denote by $\Mbar$ the set of generalized traces. If $c_i\neq\vd$ for all integers~$i$, then $\xi$ is called an \emph{infinite trace}.

The set of infinite traces is called the \emph{boundary at infinity} of the monoid~$\M$, and is denoted by~$\BM$. We observe that $\BM\neq\emptyset$ as soon as~$\Sigma\neq\emptyset$. Furthermore, $\BM$~is infinite and uncountable if and only if $\M$ is not a free commutative monoid.

Using the embedding described above of $\M$ into~$\Mbar$ (each trace $x$ corresponding to its generalized normal form), we identify $\M$ with its image in~$\Mbar$. Then $\Mbar$ decomposes as the following disjoint union: $\Mbar=\M+\BM$.

\medskip
In view of~(\ref{eq:19}), it is natural to extend the partial order on $\M$ to a partial order on $\Mbar$ by putting, for $\xi=(c_i)_{i\geq1}$ and $\zeta=(d_i)_{i\geq1}$ two generalized traces:
\begin{gather*}
  \xi\leq\zeta\text{ in $\Mbar$}\iff(\forall i=1,2\ldots\quad c_i\leq  d_i\text{ in $\C$}\,).
\end{gather*}


Recall that $\Cstar$ denotes the set of non empty cliques of the trace monoid~$\M$. For each integer $i\geq1$, we define a mapping $C_i:\BM\to\Cstar$ by putting $C_i(\omega)=c_i$ whenever $\omega=(c_i)_{i\geq1}$\,. \HPOV: $C_i(\omega)$ is the $i^{\text{th}}$~layer of the infinite
heap~$\omega$.\vspace*{-2mm}

\paragraph*{Digraph of cliques.}
\label{sec:digraph-cliques}

The digraph $(\Cstar,\to)$ is called the \emph{digraph of cliques} of the monoid. Infinite paths in this digraph correspond bijectively to infinite traces in the monoid. If one follows an infinite path in~$(\Cstar,\to)$, the infinite trace $\omega$ it corresponds to satisfies that $C_i(\omega)$ is the $i^{\text{th}}$~node visited along the path.

\begin{example}
  \label{exm:2}
We depict on Fig.~\ref{fig:jqwdkjqwd} the digraph $(\Cstar,\to)$ for the trace monoid $\M=\langle a,b,c,d\ \big|\ ac=ca,\ bd=db\rangle$.
\end{example}

\begin{figure}[!h]
  $$
  \begin{array}{c|c}
    \xymatrix{\bullet\ar@{-}[r]\labell{a}&\bullet\ar@{-}[d]\labelr{b}\\
    \bullet\ar@{-}[u]\ar@{-}[r]\labell{d}&\bullet\labelr{c}
    }\qquad\strut&\strut\qquad
\xymatrix{
    *+[F]{\strut\ b\ }\ar@{<->}[rrr]\ar@{<->}[dd]&&
    &*+[F]{\strut\ a\ }\\
    &*+[F]{\strut a\cdot c}
    \ar@{<->}[r]
    \POS!U!L\ar@{<->}[ul]!D!R\POS!U!R\ar[urr]\POS!D\ar@{<->}[drr]\POS!L\ar[dl]!U!R
    &*+[F]{\strut b\cdot d}
    \POS!U!R\ar@{<->}[ur]!D!L\POS!U!L\ar[ull]\POS!D\ar@{<->}[dll]\POS!R\ar[dr]!U!L\\
    *+[F]{\strut\ c\ }&&
    &*+[F]{\strut\ d\ }\ar@{<->}[lll]\ar@{<->}[uu]\\
    }
  \end{array}
  $$
\caption{\small \textsl{Left:} Coxeter graph of $\M=\langle  a,b,c,d\ \big|\ ac=ca,\ bd=db\rangle$.\quad\textsl{Right:} digraph of cliques for~$\M$. Arrows with a double tip stand for pairs of arrows}
  \label{fig:jqwdkjqwd}\vspace*{-4mm}
\end{figure}

\paragraph*{Probabilistic dynamics for trace monoids. Valuations and visual cylinders.}
\label{sec:prob-dynam-trace}

Textbooks on probability often start from the first non trivial probabilistic experience, namely the ``infinite sequence of tosses of a coin''~\cite{billingsley95}. Implicitly, an ``infinitely repeated probabilistic experience'' involves an infinite sequence of independent and identically distributed random variables. Mathematically, all the information is encoded into a probability measure on the space of infinite words\footnote{Recall that a \emph{\slgb}\ $\FFF$ on a set $\Omega$ is a family of subsets of $\Omega$ containing $\Omega$ and closed under complement, countable union and countable intersection. A \emph{probability measure} on $(\Omega,\FFF)$ is then a set function $\nu:\FFF\to[0,1]$ such that $\nu(\Omega)=1$ and countably additive on sequences of pairwise disjoint subsets: $(i\neq j\implies A_i\cap A_j=\emptyset)\implies \nu\bigl(\bigcup_{i\geq1}A_i\bigr)=\sum_{i\geq1}\nu(A_i)$.}, with the key feature of being \emph{memoryless}. This will guide us when looking for a generalization that would apply to trace monoids instead of word monoids.

\medskip
Let $\M$ be a non trivial trace monoid. The boundary at infinity~$\BM$ is a subset of the product set~$\Cstar^{\bbZ_{\geq1}}$ and as such, comes equipped with a topology which in turn induces a Borel \slgb, which is always understood. Assume given a probability measure $\nu$ on~$\BM$. Then basic results from measure theory show that $\nu$ is entirely determined by the countable collection of values $\nu(\up x)$, for $x$ ranging over~$\M$, where $\up x$ is the \emph{visual cylinder\footnote{The terminology \emph{visual cylinder} is derived from the ``visual measure'' introduced in geometric group theory.} of base~$x$}, defined by:
\begin{gather}
  \label{eq:12}
  \up x=\{\omega\in\BM\tqs x\leq\omega\}.
\end{gather}

\HPOV: $\omega\in\up x$ means that the infinite heap $\H \omega$ can be obtained from $\H x$ by adding infinitely many pieces from the top (again, note that the initial layers of $\H x$ and of $\H\omega$ may differ).

\medskip
We say that $\nu$ is \emph{memoryless} if it satisfies the following property:
\begin{gather}
  \label{eq:20}
  \forall x,y\in\M\quad \nu\bigl(\up(x\cdot y)\bigr)=\nu(\up x)\nu(\up y).
\end{gather}

\begin{remark}[other probabilistic dynamics]
  \label{rem:2}
In order to equip a trace monoid with a probabilistic dynamics, one might think first of the random walk approach. Consider an infinite sequence $X_1,X_2,\ldots$ of random letters, each letter $X_i$ being picked at random and uniformly in~$\Sigma$; then form the infinite trace obtained by piling up all this letters and consider the probability law of the infinite trace thus obtained. By construction, this law is indeed a probability measure on~$\BM$. Intuitively, the more a trace has internal commuting elements, the more it will be favored by this law; it is thus not a ``uniform'' way of choosing traces, and neither is it memoryless. Actually, the associated probability measure never satisfies the property~(\ref{eq:20}), as soon as $\M$ is not a free monoid nor a free commutative monoid---hence, in all cases of interest. This random walk measure is of course of deep mathematical interest on its own; yet we are looking for other probabilistic dynamics of interest.

\medskip
  Another way of selecting traces at random is the following. For each integer~$n$, consider the finite set of traces of length~$n$, and then choose randomly one trace among them. This yields a probability distribution $\nu_n$ for each integer~$n$. This procedure can be refined by attributing multiplicative weights to letters instead of the same weight to all letters. In all cases, this procedure has three drawbacks:
  \begin{enumerate}
  \itemsep=0.9pt
    \item the probability distributions $\nu_n$ are defined on $\M$ and not on~$\BM$---and it would be rather unnatural to stop the process at some fixed length~$n$;
    \item the sequence $(\nu_n)_{n\geq0}$ is not a ``consistent sequence''---hence Kolmogorov's extension theorem does not apply to define a completion ``at infinity''; and
    \item for each fixed~$n$, the probability distribution $\nu_n$ is not memoryless---if one would care to define any sort of memoryless property for probability distributions on finite traces rather than on infinite ones.
  \end{enumerate}
Despite all these restrictions, the sequence $(\nu_n)_{n\geq0}$ is of much interest since after all, it is the most natural way to pick a trace at random. It can be shown in a precise way that the sequence of probabilities $(\nu_n)_{n\geq0}$ converges to  a probability measure on~$\BM$ which is indeed memoryless. Intuitively, the memoryless probability measures that we are seeking correspond to this procedure, but obtained with ``$n=\infty$''.
\end{remark}

As observed above, the mere existence of at least a memoryless probability measure for a trace monoid is not obvious. But assuming for the moment the existence of such a probability measure~$\nu$, consider the function $f:\M\to\bbR_{\geq0}$ defined by:
\begin{gather}
  \label{eq:16}
  \forall x\in\M\quad f(x)=\nu(\up x).
\end{gather}
Then, by~(\ref{eq:20}), $f$ satisfies:
\begin{align}
  \label{eq:13}
  f(\vd)&=1&&\text{and}&\forall x,y\in\M\quad f(x\cdot y)=f(x)f(y)
\end{align}
We define a function $f:\M\to\bbR_{\geq0}$ to be a \emph{valuation} whenever it satisfies the two properties in~\eqref{eq:13}. Clearly, a valuation $f$ is entirely determined by the finite collection of its values on the letters of~$\Sigma$. And conversely, given any family $(\lambda_a)_{a\in\Sigma}$ of non negative numbers, there is a unique valuation $f$ such that $f(a)=\lambda_a$ for all $a\in\Sigma$. The central question is now the following.
\begin{quote}
\textbf{(Q)} \itshape
Let\/ $(\lambda_a)_{a\in\Sigma}$ be a collection of non negative real numbers, and let $f$ be the corresponding valuation. What computable conditions on $(\lambda)_{a\in\Sigma}$ are necessary and sufficient for the existence of a probability measure $\nu$ on $\BM$ such that\/ $\nu(\up x)=f(x)$ for all $x\in\M$?
\end{quote}

\noindent The probability measure $\nu$ thus constructed shall necessarily be memoryless. Hence answering the above question amounts to having an operational description of memoryless probability measures on~$\BM$.\vspace*{-2mm}

\paragraph*{M\"obius transform and probabilistic valuations.}
\label{sec:mobius-transform}

Our answer to the above question~\textbf{(Q)} is based on the notion of Möbius transform, a notion attached to a large class of partial orders and popularized by G.-C.~Rota~\cite{rota64}. The partial order we shall focus on is the finite partial order~$(\C,\leq)$. Let $f:\C\to A$ be any function, where $A$ is a commutative group---we shall always take $A=\bbR$. The \emph{Möbius transform} of~$f$ is the function $h:\C\to A$ defined by:
\begin{gather}
  \label{eq:1}
  \forall c\in\C\quad h(c)=\sum_{c'\in \C\tqs c\leq c'}(-1)^{|c'|-|c|}f(c').
\end{gather}

The function $f$ can be retrieved from $h$ thanks to the \emph{Möbius inversion formula}, which is a kind of generalized inclusion-exclusion formula:
\begin{gather}
  \label{eq:2}
  \forall c\in\C\quad f(c)=\sum_{c'\in\C\tqs c\leq c'}h(c').
\end{gather}
In particular, one has:
\begin{gather}
  \label{eq:3}
  f(\vd)=\sum_{c\in \C}h(c).
\end{gather}

Let $h:\C\to\bbR$ be the Möbius transform of a valuation~$f$, restricted to~$\C$. Then we define $f$ to be a \emph{probabilistic valuation} whenever:
\begin{gather}
  \label{eq:7}
\bigl(  h(\vd)=0\bigr)\quad\text{and}\quad\bigl(\forall c\in\Cstar\quad h(c)\geq0\bigr).
\end{gather}
In this case, the vector $\bigl(h(c)\bigr)_{c\in\Cstar}$ is a probability vector. Indeed, it is non negative and it sums up to $1$ thanks to~\eqref{eq:3}, since $f(\vd)=1$ and $h(\vd)=0$.

\medskip
The following statement provides an answer to Question~\textbf{(Q)}: \textit{the existence of a memoryless measure $\nu$ associated with the valuation $f$ through \mbox{$f(x)=\nu(\up x)$} for $x$ ranging over~$\M$, is equivalent to $f$ being a probabilistic valuation.}

A particular case is when $f$ is \emph{uniform}, in the sense that $f(a)=t$ is constant for $a$ ranging over~$\Sigma$, and thus $f(x)=t^{|x|}$ for $x\in\M$. A result is: {\itshape there exists a unique uniform probabilistic valuation}. It implies the existence of at least one memoryless measure for every trace monoid.

\begin{example}
\label{sec:example-1}

Let $\M=\langle a,b,c,d\tq ac=ca,\ bd=db\rangle$ with $\Cstar=\{a,b,c,d,a\cdot c,b\cdot d\}$, and whose Coxeter graph is depicted on Fig.~\ref{fig:jqwdkjqwd}. Let us simply denote by $a$, $b$, etc, the values of $f(a)$, $f(b)$, etc, for some valuation~$f$. The normalization conditions~(\ref{eq:7}) for $f$ to be a probabilistic valuation are:
\begin{gather*}
  1-a-b-c-d+ac+bd=0\\
  \begin{aligned}
    a-ac&\geq0,&b-bd&\geq0,&c-ac&\geq0,&d-db&\geq0,&
ac&\geq0,&bd&\geq0.
  \end{aligned}
\end{gather*}
A solution is to put $a=b=1/3$ and $c=d=1/4$. Another solution is  to look for the uniform valuation, hence to put $a=b=c=d=1-\sqrt2/2$. The later value is the root of smallest modulus of the polynomial $1-4p+2p^2$, which we encounter below as the Möbius polynomial of the monoid.
\end{example}\vspace*{-2mm}

\paragraph*{Markov chain of cliques.}
\label{sec:mark-chains-cliq}

Let $\nu$ be a memoryless probability measure on~$\BM$. Intuitively, the \linebreak measure $\,\nu\,$ encodes a way of choosing at random$\,$ an infinite$\,$ trace $\,\omega\in\BM.\,$ Since $\,\omega\,$ has the form
\eject

\noindent $\omega\!=\!(c_i)_{i\geq1}$,  it is natural to investigate the nature of the random sequence $c_i=C_i(\omega)$. It is random indeed since it depends on the random outcome $\omega$ of the probabilistic experience.

\medskip
It turns out that: {\itshape
  with respect to the memoryless probability measure~$\nu$, the random sequence $(C_i)_{i\geq1}$  is a homogeneous Markov chain on~$\Cstar$. Its initial distribution is given by: $\forall c\in\Cstar\quad \nu(C_1=c)=h(c)$, where $h$ is the Möbius transform of~$f$, probabilistic valuation attached to~$\nu$ as in~\eqref{eq:16}.} The transition matrix of the Markov chain can also be described, but we shall not need it in the sequel. We simply mention that it also involves the Möbius transform~$h$ (see the details in~\cite{abbes15}).\vspace*{-1mm}

\begin{remark}
  Observe the different probabilistic interpretations of the two functions $f$ and~$h$. If $c$ is some non empty clique, then $f(c)=\nu(\up c)$ is the probability that the initial clique $C_1$ \emph{contains}~$c$; whereas $h(c)$ is the probability that $C_1$ \emph{equals}~$c$. The $h(c)s$ sum up to $1$ over~$\Cstar$, whereas: $\sum_{c\in\Cstar}f(c)>1$ unless $\M$ is a free monoid.\vspace*{-1mm}
\end{remark}

\begin{example}
  Continuing with the trace monoid $\M=\langle a,b,c,d\tq ac=ca,\ bd=db\rangle$ from Example~\ref{sec:example-1}, let us determine the law $h$ of the first clique for the uniform probabilistic valuation, say~$f$, given by $f(x)=r^{|x|}$ for any $x\in\M$ with $r=1-\sqrt2/2$. The law $h$ is simply the Möbius transform of~$f$, hence:
  \begin{align*}
    h(a)=h(b)=h(c)=h(d)&=r-r^2&h(ac)=h(bd)&=r^2
  \end{align*}
Observe two things: first, the law $h$ is \emph{not} uniform on cliques; and second, the fact that all the above values sum up to~$1$ writes as: $4(r-r^2)+2r^2=1$, or equivalently as $1-4r+2r^2=0$, which holds since $r$ has been chosen as a root of the polynomial $\mu(t)=1-4t+2t^2$. This is a general fact, better explained when introducing the notion of Möbius polynomial below.\vspace*{-2mm}
\end{example}

\paragraph*{Irreducibility of trace monoids.}
\label{sec:irreducibility-1}

Given two trace monoids $\M_i=\M(\Sigma_i,I_i)$, $i=1,2$, their direct product $\M_1\oplus\M_2$ is isomorphic to another trace monoid. Indeed, take the disjoint union $\Sigma=\Sigma_1+\Sigma_2$, and for dependence relation the disjoint union $D=D_1+D_2$, with $D_1$ and $D_2$ now seen as binary relations on~$\Sigma$. Take finally the independence pair $I=(\Sigma\times\Sigma)\setminus D$. Then $\M_1\oplus\M_2$ is isomorphic to~$\M(\Sigma,I)$. In this construction, letters from a common alphabet keep their dependence relations, and letters from distinct alphabets are set to be independent, \ie, commutative.

For instance, the free commutative monoid on $N$ generators is obtained as the direct product of $N$ copies of the free monoid on $1$ generator.

Conversely, given a trace monoid $\M=\M(\Sigma,I)$, it is well known that $\M$ is not isomorphic to the product of two non trivial trace monoids if and only if the Coxeter graph $(\Sigma,D)$ is connected. In this case, the trace monoid $\M$ is said to be \emph{irreducible}.

For example, the free monoid $\M(\Sigma,I_1)$ is irreducible, and the free commutative monoid $\M(\Sigma,I_0)$ is irreducible if and only if $|\Sigma|\leq 1$. All other examples of trace monoids that we encountered previously are irreducible.\vspace*{-2mm}

\paragraph*{Combinatorics and probability for trace monoids: growth series and Möbius polynomials.}
\label{sec:growth-series-mobius}

The  \emph{growth series} $G(z)$ and the \emph{Möbius polynomial} $\mu(z)$ of a trace monoid $\M$ are defined as follows:
\begin{align}
\label{eq:18}
  G(z)&=\sum_{x\in\M}z^{|x|},&\mu(z)&=\sum_{c\in\C}(-1)^{|c|}z^{|c|}.
\end{align}
The series $G(z)$ is rational, and it is the formal inverse of the Möbius polynomial:\quad $G(z)\mu(z)=1$ (see~\cite{cartier69} for a combinatorial proof, see \cite{viennot86} for a bijective proof).

\medskip
If $\Sigma\neq\emptyset$, the Möbius polynomial has a unique root of smallest modulus (see~\cite{krob03,goldwurm00}). This root, say~$r$, is real and lies in $(0,1]$. If $\Sigma=\emptyset$, we put $r=\infty$. In all cases, the radius of convergence of $G(z)$ is~$r$.

We note that: \emph{$r\geq1$ if and only if $\M$ is free commutative}---an elementary result to be generalized when dealing with deterministic concurrent systems in Sections~\ref{sec:determ-conc-syst} and~\ref{sec:irred-determ-conc}. Indeed, the coefficients of the growth series $G(z)=\sum_{n\geq0}\lambda_nz^n$ are given by $\lambda_n=\#\{x\in\M\tqs|x|=n\}$. If $\M$ is not free commutative, then $\M$ contains the free monoid on two generators as a submonoid. Hence $\lambda_n\geq 2^n$ and thus $r\leq1/2$. Whereas, if $\M$ is free commutative and $\Sigma$ has $N\geq0$ elements, then\footnote{This is a particular case of the easily observed identity on Möbius polynomials: $\mu_{\M_1\oplus\M_2}(z)=\mu_{\M_1}(z)\mu_{\M_2}(z)$.} $\mu(z)=(1-z)^N$ and therefore $r=1$ or $r=\infty$. In this case, one recovers  from the formula $G(z)=1/(1-z)^N$ the standard elementary result that free commutative monoids have a polynomial growth.


\medskip
Returning to the case of a general trace monoid~$\M$, let $f_z$ be the uniform valuation on $\M$ defined by $f_z(a)=z$ for all $a\in\Sigma$, and let $h_z$ be the Möbius transform of~$f_z$. Then, comparing~(\ref{eq:18}) with~(\ref{eq:1}), one sees that $\mu(z)=h_z(\vd)$. Therefore, for the uniform valuation $f_z$ to be probabilistic, it is necessary that $z$ is a root of the Möbius polynomial~$\mu(z)$. Actually, the following result holds if $\Sigma\neq\emptyset$, making more precise the statement introduced before: {\itshape the only value for the uniform valuation $f_z$ to be probabilistic is $z=r$, the root of smallest modulus of the Möbius polynomial~$\mu(z)$}. In other words, among the roots of the Möbius polynomial, only for the root $r$ of smallest modulus does the required condition $(\forall c\in\Cstar\quad h_r(c)\geq0)$ from~\eqref{eq:7} hold. The associated probability measure on $\BM$ is the \emph{uniform measure}. Intuitively, the uniform measure gives equal weight to all infinite traces.\vspace*{-2mm}

\paragraph*{Null nodes for trace monoids.}
\label{sec:null-nodes-trace}

Consider the uniform measure $\nu$  on the boundary at infinity $\BM$ of a trace monoid~$\M$, the associated valuation $f$ and its Möbius transform~$h$. Say that a node $c$ in the digraph of cliques $(\Cstar,\to)$ is a \emph{null node} if $h(c)=0$.

Consider also the Markov chain of cliques $(C_i)_{i\geq1}$ associated with a infinite trace~$\omega$, drawn at random according to~$\nu$. Recall that the initial distribution of the chain, hence the probability distribution of the clique~$C_1$, is given by the probability vector~$(h(c))_{c\in\Cstar}$. In particular, if $c$ is a null node, it has zero probability of being visited by the first clique of~$\omega$. Actually, it then has zero probability of being ever visited by any of the cliques~$C_i$, $i\geq1$ (this follows from the form of the transition matrix of the chain). It is thus interesting to determine the null nodes, whenever they exist.

\begin{example}
Consider the \emph{non irreducible} trace monoid $\M=\langle a,b,c\ \big|\ ac=ca,\ bc=cb\rangle$, with cliques $\C=\{\vd,a,b,c,a\cdot c,b\cdot c\}$. Let $f(x)=r^{|x|}$ be a uniform valuation. The Möbius transform of $f$ is given by: \vspace*{-2mm}
\begin{align*}
  h(\vd)&=(1-2r)(1-r)&h(a)&=r(1-r)&h(b)&=r(1-r)\\
  h(c)&=r(1-2r)&h(a\cdot c)&=r^2&h(b\cdot c)&=r^2
\end{align*}

Hence the value of $r$ corresponding to the uniform measure is $r=\frac12$, inducing the null node $c$ since then $h(c)=0$.
\end{example}

The previous example involved a non irreducible trace monoid. By contrast, the following result holds: {\itshape If $\M$ is a non trivial and irreducible trace monoid, there is no null node in $(\Cstar,\to)$}. As a consequence, one can prove that, if $\M$ is irreducible and non trivial, an infinite trace drawn at random has probability $1$ to visit infinitely often any clique.

\begin{remark}
  Null nodes, which only occur for non irreducible trace monoids, are closely related to the combinatorics of the trace monoid. A null node is a clique that has exactly zero probability to appear as the first clique of an infinite trace. When considering large traces rather than infinite traces, it can be reformulated as follows: if $c$ is a null node, then among traces of size~$n$, very few have their first clique equals to~$c$, compared to others; and the larger~$n$, the smaller this ratio. At the limit, the ratio equals zero.
\end{remark}

\subsection{Concurrent systems and their combinatorics}
\label{sec:conc-syst-their}

A natural way to generalize the notion of trace monoid is to add a notion of state. This yields the notion of concurrent system, introduced below. The background material presented in this section is borrowed from~\cite{abbes19:_markov,abbes20}. \vspace*{-2mm}

\paragraph*{Concurrent systems and executions.}
\label{sec:basics}

A \emph{concurrent system} is a triple $\X=(\M,X,\bot)$ as follows: $\M$~is a trace monoid; $X$~is a finite set of \emph{states};  $\bot$~is a special symbol not in~$X$; furthermore, we put $X'=X\cup\{\bot\}$ and there is mapping $X'\times\M\to X'$, denoted by $(\alpha,x)\mapsto\alpha\cdot x$ and satisfying the three following properties:
\begin{gather}
  \label{eq:14}
  \forall \alpha\in X'\quad \alpha\cdot\vd=\alpha\\
  \label{eq:15}
  \forall \alpha\in X'\quad \forall (x,y)\in \M\times \M\quad \alpha\cdot(x\cdot y)=(\alpha\cdot x)\cdot y\\
  \label{eq:17}\forall x\in\M\quad\bot\cdot x=\bot
\end{gather}
The properties~(\ref{eq:14}) and (\ref{eq:15}) are the axioms of a \emph{right monoid action} of $\M$ on $X\cup\{\bot\}$. As witnessed by~(\ref{eq:17}), the symbol $\bot$ represents a sink state, intended to materialize a forbidden state. So we are interested, for every $\alpha,\beta\in X$, in the following subsets of~$\M$:
\begin{align*}
  \M_{\alpha,\beta}&=\{x\in\M\tqs \alpha\cdot x=\beta\},&\M_\alpha&=\{x\in\M\tqs \alpha\cdot x\neq\bot\}.
\end{align*}
Traces of $\M_\alpha$ are called \emph{executions starting from~$\alpha$}, or \emph{executions} for short. Note that $\M_\alpha$ is always downward closed in $(\M,\leq)$, thanks to~(\ref{eq:17}).

\medskip
The concurrent system $\X$ is \emph{trivial} if $\alpha\cdot a=\bot$ for all $\alpha\in X$ and for all $a\in\Sigma$. It is \emph{non trivial} otherwise.

Borrowing the terminology from the theory of group actions, we say that the concurrent system is \emph{homogeneous} if $\M_{\alpha,\beta}\neq\emptyset$ for all pairs $(\alpha,\beta)\in X\times X$ (the state space has only one connected component). Borrowing the terminology from Petri nets theory, we say that the system is \emph{alive} if for every state $\alpha\in X$ and for every letter $a\in\Sigma$, there exists an execution $x\in\M_\alpha$ such that $a\in x$.

Finally we say that the concurrent system $\X=(\M,X,\bot)$ is \emph{irreducible} if it is non trivial, homogeneous and alive, and if $\M$ is an irreducible trace monoid. The interest of this notion of irreducibility lies in the spectral property that is stated later.

\paragraph*{Representation of concurrent systems and of executions.}
\label{sec:repr-conc-syst}

To represent a concurrent system $\X=(\M,X,\bot)$, we first use the Coxeter graph of~$\M$ already introduced for trace monoids. We also depict the \emph{labeled multigraph of states}, or \emph{graph of states} for short, whose vertices are the elements of~$X$, and with an edge from $\alpha$ to~$\beta$ labeled by the letter $a\in\Sigma$ if $\alpha\cdot a=\beta$.

\begin{remark}
  Any multigraph $V$ with edges labeled by elements from a set $\Sigma$ represents a ``next state function'', and thus extends to an action of the free monoid $(V\cup\{\bot\})\times\Sigma^*\to(V\cup\{\bot\})$, provided that for any node $v\in V$, there are no two edges starting from $v$ and labeled with the same letter. It requires an additional verification to check that it also represents an action of a trace monoid $\M=\M(\Sigma,I)$ on~$V\cup\{\bot\}$; namely, one has to check that $\alpha\cdot(ab)=\alpha\cdot(ba)$ for any pair $(a,b)\in I$ and for every vertex $\alpha\in V$. In other words, each commuting pair in $\Sigma$ must correspond to a diamond shape for every vertex in the graph of states.
\end{remark}

\begin{example}
  \label{exm:3}
Let four slots numbered $0$, $1$, $2$ and~$3$ be put in circle. Each slot stores the value $0$ or~$1$, with $0$ as initial value. If two neighboring slots store the same value, then a piece can be played: piece $a$ for slots $0$ and~$1$, piece $b$ for slots $1$ and~$2$, piece $c$ for slots $2$ and~$3$, piece $d$ for slots $3$ and~$0$. In case the piece is played, the common value of the two neighboring slots is changed to its opposite. This small game corresponds to the concurrent system $\X=(\M,X,\bot)$ with $\M=\langle a,b,c,d\ \big|\ ac=ca,\ bd=db\rangle$ (see the Coxeter graph depicted on Fig.~\ref{fig:jqwdkjqwd}), and with the set of all possible reachable configurations for the four slots as set of states. Hence, $X=\{0000,\ 1100,\ 0110,\ 0011,\ 1001,\ 1111\}$. Note that each action is \emph{reversible}: $\alpha\cdot x\cdot x=\alpha$ for all $\alpha\in X$ and all $x\in\Sigma$. The graph of states is depicted on Fig.~\ref{fig:pkazpoaaaz}.

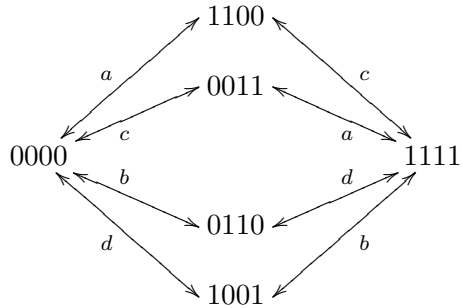
\begin{figure}[!h]
  $$
  \xymatrix@R=1.2em@C=4.2em{
    &1100&\\
    &0011\\
    0000\ar@{<->}[ruu]!L^{a}\ar@{<->}[ru]!L_{c}\ar@{<->}[rd]!L^{b}\ar@{<->}[rdd]!L_{d}
    &&1111\ar@{<->}[luu]!R_{c}\ar@{<->}[lu]!R^{a}\ar@{<->}[ld]!R_{d}\ar@{<->}[ldd]!R^{b}
\\
    &0110\\
    &1001
    }
  $$
  \caption{\small Graph of states for a concurrent system with $6$ states and associated with the trace monoid $\langle a,b,c,d\ \big|\ ac=ca,\ bd=bd\rangle$ (see the Coxeter graph of this monoid on Fig.~\ref{fig:jqwdkjqwd})
}
  \label{fig:pkazpoaaaz}
\end{figure}
\end{example}

\begin{example}
\label{exm:4}
Consider the $1$-safe Petri net depicted in Fig.~\ref{fig:petrinetrs},~$(a)$. The set of states is the set of reachable markings. The underlying trace monoid is generated by the transitions of the net, with commutative transitions $t$ and $t'$ whenever $^{\bullet}t^\bullet\cap{}^\bullet{t'}^\bullet=\emptyset$, thus $\M=\langle a,b,c,d\;|\; ad=da,\ db=db\rangle$.

\medskip
The corresponding Coxeter graph is depicted on Fig.~\ref{fig:petrinetrs},~$(b)$, and the graph of states is depicted on Fig.~\ref{fig:petrinetrs},~$(c)$. If not familiar with the model of Petri nets, the reader can ignore the picture of Fig.~\ref{fig:petrinetrs},~$(a)$, and simply focus on the Coxeter graph and the graph of states.

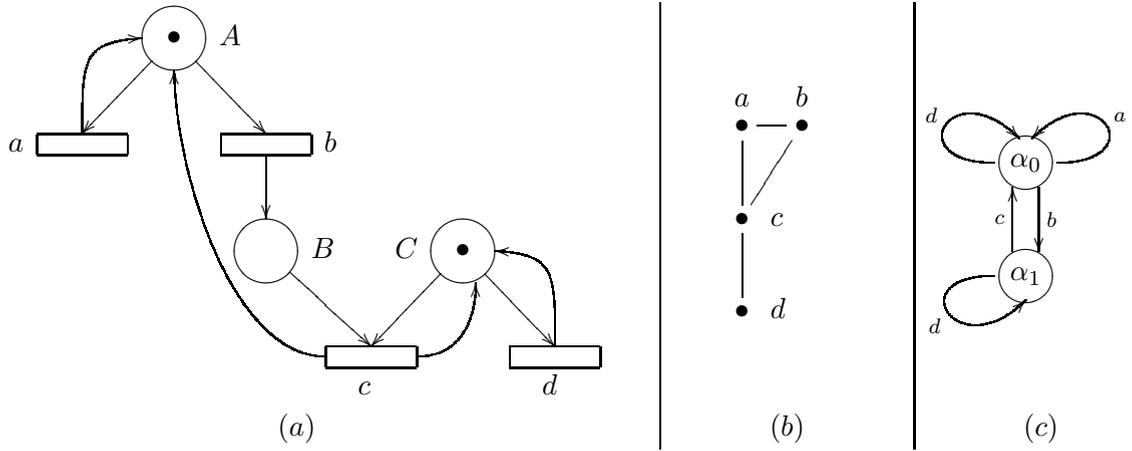
\begin{figure}[!ht]
  \centering
 \hspace*{-5mm} \begin{tabular}{c|c|c}
\begin{minipage}[c]{.6\textwidth}
  $$
  \xymatrix@C=.5em@M=.2em{
    &\place\markp\labelr{A}\POS[]\ar[dl]!U\ar[dr]!U\\
    \transh\ar@(u,l)[ur]\labell{a}&&\transh\ar[d]\labelr{b}\\
&    &\place\ar[dr]!U\labelr{B}&&\place\ar[dl]!U\ar[dr]!U!U\POS[]\markp\labell{C}\\
 &   &&\transh\POS!R\ar@(r,u)[ur]!D\POS[]\POS!L\ar@(l,d)[lluuu]!D\labeld{\strut c}&&\transh\ar@(u,r)[ul]\labeld{d}
    }
    $$
  \end{minipage}
  \hspace*{-5mm}   & \hspace*{-5mm}
       \begin{minipage}[c]{.2\textwidth}
         $$
         \xymatrix@C=1em{
           \bullet\ar@{-}[r]\ar@{-}[d]\labelu{a}&\bullet\ar@{-}[dl]\labelu{b}\\
           \bullet\ar@{-}[d]\labelr{c}\\
           \bullet\labelr{d}
}
$$
\end{minipage}
 &  \hspace*{-5mm}    \begin{minipage}[c]{.2\textwidth}
    $$
\xymatrix{*++[o][F-]{\alpha_0}%
\POS!L\ar@(l,ul)^{d}[]!U!L(.2)%
\POS!R\ar@(r,ur)_{a}[]!U!R(.2)%
\POS[]\ar@<1ex>^{b}[d]\ar@<1ex>^{c}[d];[]\\
  *++[o][F-]{\alpha_1}
\POS!L\ar@(l,dl)_{d}[]!D
    }
    $$
  \end{minipage}
    \\
    $(a)$&$(b)$&$(c)$\\
  \end{tabular}
  \caption{\small $(a)$---A safe Petri net with its initial marking $\alpha_0=\{A,C\}$ depicted. The two reachable markings are  $\alpha_0$ and $\alpha_1=\{B,C\}$.\quad$(b)$---The Coxeter graph of the associated trace monoid.\quad $(c)$---Graph of markings of the net}
  \label{fig:petrinetrs}
\end{figure}
\end{example}

\vspace*{-5mm}
\paragraph*{Notations, generalized and infinite executions.}
\label{sec:notat-runn-exampl}

Given a concurrent system $\X=(\M,X,\bot)$, we introduce the following notations, for $\alpha,\beta\in X$:
\begin{align*} \Sigma_\alpha&=\Sigma\cap\M_\alpha&\C_\alpha&=\C\cap\M_\alpha&
\Cstar_\alpha&=\Cstar\cap\M_\alpha&                                                                              \C_{\alpha,\beta}&=\C\cap\M_{\alpha,\beta}
\end{align*}

A \emph{generalized execution from~$\alpha$} is a generalized trace $\xi\in\Mbar$ such that:
\begin{gather*}
  \forall x\in \M\quad x\leq\xi\implies x\in\M_\alpha.
\end{gather*}
Their set is denoted~$\Mbar_\alpha$, and we also put $\BM_\alpha=\Mbar_\alpha\cap\BM$. Elements of $\BM_\alpha$ represent infinite executions of the system starting from the initial state~$\alpha$. Note that, even for a non trivial concurrent system, some sets or even all sets $\BM_\alpha$ might be empty, which contrasts with the situation for trace monoids.

Every trace monoid $\M$ can be seen as a concurrent system with a single state by considering $\X=(\M,X,\bot)$ with $X=\{*\}$, and $*\cdot x=*$ and $\bot\cdot x=\bot$ for every $x\in\M$. It is then irreducible as a concurrent system if and only if $\M$ is non trivial and irreducible as a trace monoid.\hspace*{-2mm}

\paragraph*{Digraph of states-and-cliques.}
\label{sec:digr-stat-cliq}

Infinite executions of a concurrent system $\X=(\M,X,\bot)$ are, in particular, infinite traces of~$\M$. As seen in Sect.~\ref{sec:generalised-traces}, infinite traces correspond to paths in the digraph of cliques $(\Cstar,\to)$. Not all infinite paths of $(\Cstar,\to)$ however correspond, in general, to infinite executions of~$\X$. In order to take into account the constraints induced by the monoid action, we introduce the \emph{digraph of states-and-cliques} $(\Dstar,\to)$, the vertices of which are pairs $(\alpha,c)$ with $\alpha$ ranging over $X$ and $c$ ranging over~$\Cstar_\alpha$. There is an arrow $(\alpha,c)\to(\beta,d)$ in $\Dstar$ if $\beta=\alpha\cdot c$ and if $(c,d)$ is a normal pair of cliques.

\medskip To every infinite execution $\omega=(c_i)_{i\geq1}$ from~$\alpha$, is associated the infinite path $(\alpha_{i-1},c_i)_{i\geq1}$ in~$\Dstar$, where $\alpha_i$ is defined by $\alpha_0=\alpha$ and $\alpha_i=\alpha\cdot(c_1\cdot\ldots\cdot c_i)$ for $i\geq1$. We put $Y_i(\omega)=(\alpha_{i-1},c_i)$ for every integer $i\geq1$. This is the $i^\text{th}$~``state-and-clique'' of the system, when the infinite execution $\omega$ is scanned according to its normal form. Conversely, every infinite path in $\Dstar$ corresponds to a unique infinite execution.

\begin{example}
For the Petri net of Fig.~\ref{fig:petrinetrs}, the digraph of states-and-cliques $\Dstar$ is depicted on Fig.~\ref{fig:pokqwdoijjq}. Here is how to obtain it ``by hand''. For every state~$\alpha$, compute first the sub-alphabet $\Sigma_\alpha=\{a\in\Sigma\tqs\alpha\cdot a\neq\bot\}$. Here, $\Sigma_{\alpha_0}=\{a,b,d\}$ and $\Sigma_{\alpha_1}=\{c,d\}$.
 Comparing with the Coxeter graph of the monoid, keep note of all the cliques that can be formed using only letters from~$\Sigma_\alpha$, and retain from these only the non empty cliques $\gamma$ such that $\alpha\cdot\gamma\neq\bot$. Their set is~$\Cstar_\alpha$. Here, $\Cstar_{\alpha_0}=\{a,b,d,a\cdot d,b\cdot d\}$ and $\Cstar_{\alpha_1}=\{c,d\}$. Then for every state $\alpha$ and for every  $\gamma\in\Cstar_\alpha$, compute $\beta=\alpha\cdot\gamma$ on the one hand, and all $\delta\in\Cstar_\beta$ such that $\gamma\to\delta$ holds on the other hand. The pairs $(\beta,\delta)$ thus obtained are the successors of $(\alpha,\gamma)$ in~$\Dstar$.

 \begin{figure}[!h]
     $$
     \xymatrix@C=1em{
     & &&*+[F]{(\alpha_0,d)}\POS!R!D(.5)\ar@(dr,ur)[]!R!U(0.5)&&
\\     &*+[F]{(\alpha_0,ad)}\POS!U\ar[urr]!L\POS[]\ar[rrr]\POS!L\ar[dl]!U\POS[]\ar[d]\POS[]\POS!L!D(0.5)\ar@(dl,dr)[]!D!L(0.5)
      &&&*+[F]{(\alpha_0,bd)}\POS!D\ar[dl]!R\ar[ddl]!R&\\
      *+[F]{(\alpha_0,a)}\ar[r]\POS!L!D(0.5)\ar@(dl,dr)[]!D!L(0.5)
      &*+[F]{(\alpha_0,b)}
      &&*+[F]{(\alpha_1,c)}\POS!U!R(0.25)\ar[ur]!D!L\POS[]\ar@<1ex>[ll]\ar@<1ex>[ll];[]\POS!L!U\ar[ull]!R!D
      \POS[]\ar'[u][uu]\POS!D!L(0.5)\ar@(dl,dr)[lll]!D!R(0.5)
\\
&&&*+[F]{(\alpha_1,d)}\POS[]\POS!L!U(0.5)\ar@(ul,dl)[]!L!D(0.5)\POS[]\ar[u]&
       }
       $$
   \caption{\small Digraph of states-and-cliques for the concurrent system corresponding to the Petri net depicted on Fig.~\ref{fig:petrinetrs}}
   \label{fig:pokqwdoijjq}
 \end{figure}
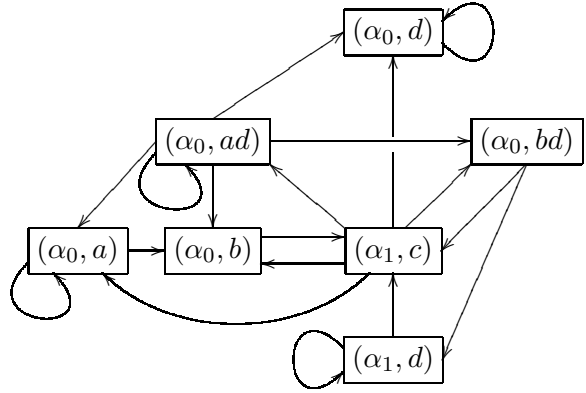
  \end{example}

\vspace*{-4mm}
\paragraph*{Valuations and probabilistic valuations. Markov concurrent measures.}
\label{sec:valu-prob-valu}

In this section we extend to concurrent systems the notions of valuations and of probabilistic valuations introduced earlier for trace monoids.

\medskip
A \emph{valuation} on a concurrent system $\X=(\M,X,\bot)$ is a family $f=(f_\alpha)_{\alpha\in X}$ of mappings $f_\alpha:\M\to\bbR_{\geq0}$ satisfying the three following properties:
\begin{gather}
  \label{eq:4}
  \forall \alpha\in X\quad\forall x\in\M\quad \alpha\cdot x=\bot\implies f_\alpha(x)=0\\
  \label{eq:5}
  \forall\alpha\in X\quad\forall x\in\M_\alpha\quad\forall y\in\M_{\alpha\cdot x}\quad f_{\alpha}(x\cdot y)=f_\alpha(x)f_{\alpha\cdot x}(y)\\
  \label{eq:6}
  \forall\alpha\in X\quad f_\alpha(\vd)=1
\end{gather}

Let $f=(f_\alpha)_{\alpha\in X}$ be a valuation and for each $\alpha\in X$, let $h_\alpha:\C\to\bbR$ be the Möbius transform of the restriction $f_\alpha\rest\C:\C\to\bbR_{\geq0}$. Note that $h_\alpha(x)=0$ if $x\notin\M_\alpha$. We say that $f$ is a \emph{probabilistic valuation} if, for every state $\alpha\in X$:
\begin{gather}
  \label{eq:8}
h_\alpha(\vd)=0\quad\text{and}\quad\bigl(\forall c\in\Cstar_\alpha\quad h_\alpha(c)\geq0\bigr)
\end{gather}

In this case, there exists a unique family $\nu=(\nu_\alpha)_{\alpha\in X}$, where $\nu_\alpha$ is a probability measure on~$\BM_\alpha$, such that $\nu_\alpha(\up x)=f_\alpha(x)$ for all $\alpha\in X$ and for all $x\in\M_\alpha$. Of course the existence of a probabilistic valuation implies in particular that $\BM_\alpha\neq\emptyset$ for all $\alpha\in X$. Such a family $(\nu_\alpha)_{\alpha\in X}$ is called a \emph{Markov concurrent measure}, because of the chain rule~(\ref{eq:5}) which is reminiscent of the classical property of Markov chains. For the Markov measure~$(\nu_\alpha)_{\alpha\in X}$, the chain rule reads as follows:
\begin{gather}
  \label{eq:21}
  \forall\alpha\in X\quad\forall x\in\M_\alpha\quad\forall y\in\M_{\alpha\cdot x}\quad
  \nu_\alpha\bigl(\up(x\cdot y)\bigr)=\nu_\alpha(\up x)\nu_{\alpha\cdot x}(\up y)
\end{gather}
The chain rule~\eqref{eq:21} extends to concurrent systems the memoryless property~\eqref{eq:20} for trace monoids.\hspace*{-2mm}

\paragraph*{ Markov chain of states-and-cliques.}
\label{sec:markov-chain-states}

If $\nu=(\nu_\alpha)_{\alpha\in X}$ is associated as above with a probabilistic valuation $f=(f_\alpha)_{\alpha\in X}$, then for each state $\alpha\in X$, and with respect to the probability measure~$\nu_\alpha$\,, the family of mappings $Y_i:\BM_\alpha\to\Dstar$ defined earlier is a homogeneous Markov chain, called the \emph{Markov chain of states-and-cliques}. Its initial distribution is given by $\un_{\alpha}\otimes h_\alpha$, meaning:
\begin{gather}
  \label{eq:9}
\forall\alpha\in X\quad\forall c\in\Cstar_\alpha\quad  \nu_\alpha(C_1=c)=h_\alpha(c).
\end{gather}
In other words, even though the ``user'' may choose the initial state $\alpha\in X$ of the system, this ``user'' does not have control on the initial \emph{clique} of a random infinite execution of the system starting from~$\alpha$. Indeed, the probability distribution of the first clique is precisely
given by~\eqref{eq:9}.\hspace*{-2mm}

\paragraph*{Example and null nodes.}
\label{sec:example-null-nodes}

Let us determine the probabilistic valuations for the Petri net example of Fig.~\ref{fig:petrinetrs}, which is an irreducible concurrent system. Any probabilistic valuation $f=(f_\alpha)_{\alpha\in X}$ is entirely determined by the \emph{finite} family of values $f_\alpha(u)$ for $(\alpha,u)$ ranging over $\{\alpha_0,\alpha_1\}\times \Sigma$, since then the other values $f_\alpha (x)$ are obtained by the chain rule $f_\alpha(x\cdot y)=f_\alpha(x)f_{\alpha\cdot x}(y)$.

\medskip
Since $f_{\alpha_0}(c)=f_{\alpha_1}(a)=f_{\alpha_1}(b)=0$, the remaining parameters for $f$ are $p=f_{\alpha_0}(a)$, $q=f_{\alpha_0}(b)$, $s=f_{\alpha_0}(d)$, $t=f_{\alpha_1}(c)$, $u=f_{\alpha_1}(d)$. The parameters are not independent; for coherence with the commutativity relations induced by the trace monoid, one must have $f_{\alpha_0}(a)f_{\alpha_0\cdot a}(d)=f_{\alpha_0}(d)f_{\alpha_0\cdot d}(a)$,  since $a\cdot d=d\cdot a$, and $f_{\alpha_0}(b)f_{\alpha_0\cdot b}(d)=f_{\alpha_0}(d)f_{\alpha_0\cdot d}(b)$ since $b\cdot d=d\cdot b$; yielding simply
$qu=qs$ here.

\begin{table}[!h]
  \centering\small
   \caption{\small Möbius transform of a generic valuation for the Petri net example depicted in Fig.~\ref{fig:petrinetrs}, with parameters
  $p=f_{\alpha_0} (a)$, $q=f_{\alpha_0} (b)$, $s=f_{\alpha_0} (d)=f_{\alpha_1}(d)$  and $t=f_{\alpha_1}(c)$}
  \label{tab:moniuadsa}
  \scalebox{0.93}{
  \begin{tabular}{c|c|c|c|c|c|c|c}
     {\rm state} $\alpha$  & $h_\alpha(\varepsilon)$ & $h_\alpha(a)$ & $h_\alpha(b)$ & $h_\alpha(c)$ & $h_\alpha(d)$ & $h_\alpha(a\cdot d)$ & $h_\alpha(b\cdot d)$\\
      \hline
      $\alpha_0$ & $1-p-q-s+ps+qs$
                &$p-ps$ & $q-qs$ & $0$ & $s-ps-qs$ & $ps$ & $qs$\\[1.8em]
      $\alpha_1$ & $1-t-s$
           &$0$ & $0$ & $t$ & $s$ &$0$ &$0$\\
    \end{tabular}}
 \end{table}

To simplify the exposition, we eliminate the border cases and restrict our attention to the case where all parameters stay within the open interval $(0,1)$. Then we obtain $u=s$ from the previous equality $qu=qs$.

The Möbius transform of $f_{\alpha_0}$ evaluated for instance at $b$ is $h_{\alpha_0}(b)=f_{\alpha_0}(b)-f_{\alpha_0}(bd)=f_{\alpha_0}(b)-f_{\alpha_0}(b)f_{\alpha_1}(d)=q-qs$. Other computations are done similarly, and we gather the results in Table~\ref{tab:moniuadsa}.
According to~(\ref{eq:8}), the normalization constraints on the parameters for the valuation $f$ to be probabilistic are thus:
\begin{align}
  \label{eq:10}
  h_{\alpha_0}(\vd)=0&:\  1-p-q-s+ps+qs=0\\
\label{eq:22}  h_{\alpha_1}(\vd)=0&:\
                1-t-s=0,
\end{align}
plus all inequalities $h_{\alpha_0}(a)\geq0$, etc, which in this case do not bring any additional constraints.

\medskip
Here, the equation in~\eqref{eq:10} rewrites as $(1-p-q)(1-s)=0$. It follows that $1-p-q=0$ and, in view of Table~\ref{tab:moniuadsa}, it implies $h_{\alpha_0}(d)=0$. This illustrates the notion of null node for concurrent systems.

\medskip
We define a node $(\alpha,c)$ of $\Dstar$ to be a \emph{null node}, with respect to some probabilistic valuation $f=(f_\alpha)_{\alpha\in X}$, if $h_\alpha(c)=0$, where $h_\alpha$ is the Möbius transform of~$f_\alpha$. As for trace monoids, null nodes are never reached by the Markov chain of states-and-cliques. But, contrasting with the case of trace monoids, null nodes may exist even for irreducible concurrent systems, as the previous example shows. \hspace*{-2mm}

\paragraph*{Characteristic root of a concurrent system.}
\label{sec:characteristic-root}\label{sec:irreducibility}

Consider the \emph{Möbius matrix} \\
$\mu(z)=(\mu_{\alpha,\beta}(z))_{(\alpha,\beta)\in X\times X}$, the polynomial~$\theta(z)$ with integer coefficients, and the \emph{growth matrix} $G(z)=(G_{\alpha,\beta}(z))_{(\alpha,\beta)\in X\times X}$ defined by:
\begin{align*}
  \mu_{\alpha,\beta}(z)&=\sum_{c\in\C_{\alpha,\beta}}(-1)^{|c|}z^{|c|}&\theta(z)&=\det\mu(z)&G_{\alpha,\beta}(z)=\sum_{x\in\M_{\alpha,\beta}}z^{|x|}
\end{align*}
Then $G(z)$ is a matrix of rational series, and it is the inverse of the Möbius matrix: $G(z)\mu(z)=\text{Id}$. One of the roots of smallest modulus of the polynomial $\theta(z)$ is real and lies in $(0,1]\cup\{\infty\}$, with the convention that it is $\infty$ if $\theta(z)$ is a non zero constant. By definition, this non negative real or $\infty$ is the \emph{characteristic root} of the concurrent system~$\X$. The characteristic root $r$ coincides with the minimum of all convergence radii of the generating series $G_{\alpha,\beta}(z)$, for $(\alpha,\beta)$ ranging over~$X\times X$. Intuitively, the smaller is~$r$, the ``bigger'' is~$\X$, in the sense of a large set of executions.

\begin{example}
  \label{exm:6}
For the Petri net example from Fig.~\ref{fig:petrinetrs}, the Möbius matrix is given by: \begin{gather*}
  \mu(z)=
  \begin{array}{c}
    \alpha_0\\\alpha_1
  \end{array}
  \begin{pmatrix}
    1-2z+z^2&-z+z^2\\
    -z&1-z
  \end{pmatrix}
\end{gather*}
with determinant $\theta(z)=(1-z)^2(1-2z)$. The characteristic root is $r=1/2$.
\end{example}

\begin{example}
  \label{exm:7}
For the concurrent system from Example~\ref{exm:3}, whose graph of states is depicted in Fig.~\ref{fig:pkazpoaaaz}, the Möbius matrix is the following.

\eject
\hbox{}
\vspace*{-11mm}

\begin{gather*}
  M(z)=
  \begin{array}{c}
    0000\\
    1100\\
    0011\\
    0110\\
    1001\\
    1111
  \end{array}
  \begin{pmatrix}
    1&-z&-z&-z&-z&2z^2\\
    -z&1&z^2&0&0&-z\\
    -z&z^2&1&0&0&-z\\
    -z&0&0&1&z^2&-z\\
    -z&0&0&z^2&1&-z\\
    2z^2&-z&-z&-z&-z&1\\ \vspace*{-6mm}
  \end{pmatrix}
\end{gather*}
\end{example}

\paragraph*{The spectral property for irreducible concurrent systems.}
\label{sec:spectr-prop-irred}

Consider a concurrent system $\X=(\M,X,\bot)$, with $\M=\M(\Sigma,I)$. If $\Sigma'$ is any subset of~$\Sigma$, and if $\M'=\langle\Sigma'\rangle$ is the submonoid of $\M$ generated by~$\Sigma'$, which is indeed a trace monoid, then the restriction of the action $(X\cup\{\bot\})\times\M'\to X\cup\{\bot\}$ defines a new concurrent system $\X'=(\M',X,\bot)$, said to be \emph{induced by restriction}. In particular, let $\X^a$ denote the concurrent system induced by restriction with $\Sigma'=\Sigma\setminus\{a\}$, and let $r^a$ be the characteristic root of~$\X^a$.

A key property, that we shall use later, is the \emph{spectral property}~\cite{abbes20} which states:\quad {\itshape if $\X$ is irreducible, then $r^a>r$ for every  $a\in\Sigma$}. The point here is the strict inequality, which derives from the irreducibility of~$\X$; indeed, the inequality $r^a\geq r$ is always valid without restriction on~$\X$.\hspace*{-2mm}

\paragraph*{Uniform measure for concurrent systems.}
\label{sec:unif-meas-conc}

We have seen in Sect.~\ref{sec:growth-series-mobius} the existence of a particular probability measure on the boundary at infinity of every trace monoid~$\M(\Sigma,I)$, namely the uniform measure, associated with the unique uniform and probabilistic valuation. The uniform valuation was defined by $f(x)=r^{|x|}$, where $r$ is the root of smallest modulus of the Möbius polynomial of $(\Sigma,I)$.

For concurrent systems, an analogous notion exists in most cases, and in particular if the system is irreducible. Say that a mapping $\Gamma:X\times X\to\bbR_{>0}$ is a \emph{cocycle} whenever it satisfies $\Gamma(\alpha,\gamma)=\Gamma(\alpha,\beta)\Gamma(\beta,\gamma)$ for all triples $(\alpha,\beta,\gamma)\in X^ 3$. The following result holds: {\itshape If $\X=(\M,X,\bot)$ is an irreducible concurrent system, there exists a unique probabilistic valuation of the form $f_\alpha(x)=t^{|x|}\Gamma(\alpha,\alpha\cdot x)$, for $x\in\M_\alpha$, where $t$ is a positive real and\/ $\Gamma:X\times X\to\bbR_{>0}$ is a cocycle. The real $t$ is the characteristic root of~$\X$, and the cocycle $\Gamma$ is called the \emph{Parry cocycle}}. The associated concurrent Markov measure $(\nu_\alpha)_{\alpha\in X}$ is the \emph{uniform measure} of~$\X$.

The Parry cocycle has a combinatorial interpretation on which additional details are given in~\cite{abbes19:_markov}. It can be determined as follows. Let $\mu(r)$ be the Möbius matrix of the system evaluated at~$r$, characteristic root of~$\X$. Then, by definition of~$r$, $\mu(r)$~has a non trivial kernel. It actually holds that $\dim\ker\mu(r)=1$. Hence, let $(v_\alpha)_{\alpha\in X}$ be a non zero element of $\ker\mu(r)$. Then $\Gamma(\cdot,\cdot)$ is given by $\Gamma(\alpha,\beta)=v_\beta/v_\alpha$, and it holds indeed that $(v_\alpha)_{\alpha\in X}$ has all its coordinates non zero.

\begin{example}
Let us determine the uniform measure for the Petri net example of Fig.~\ref{fig:petrinetrs}. According to the computation already done in Example~\ref{exm:6}, the Möbius matrix evaluated at the characteristic root $r=\frac12$ is
$\mu(\frac12)=\left(
  \begin{smallmatrix}
    \frac14 &-\frac14\\-\frac12&\frac12
  \end{smallmatrix}\right)
$. A non zero vector of its kernel is~$\left(\begin{smallmatrix}1\\1  \end{smallmatrix}\right)$,
hence the Parry cocycle is constant equal to~$1$. The uniform probabilistic valuation is thus $f_\alpha(x)=\bigl(\frac12\bigr)^{|x|}$. For a double check, we can verify that the two conditions stated earlier in~\eqref{eq:10} and~\eqref{eq:22} for this example are satisfied by this valuation (the corresponding values are $p=q=s=t=\frac12$). The probability law of the first clique when starting from a state $\alpha$ is given by the Möbius transform~$h_\alpha$ of~$f_\alpha$.  So for instance the probability law of the first clique when starting from $\alpha_0$ is given by:
\begin{align*}
  \nu_{\alpha_0}(C_1=a)&=\frac14
  &\nu_{\alpha_0}(C_1=b)&=\frac14\\
  \nu_{\alpha_0}(C_1=a\cdot d)&=\frac14
  &\nu_{\alpha_0}(C_1=b\cdot d)&=\frac14
  &\nu_{\alpha_0}(C_1=d)&=0
\end{align*}
As already observed, the node $(\alpha_0,d)$ is a null node. The first clique of a random infinite execution starting from $\alpha_0$ has probability $0$ to be~$d$;  although it was not impossible \emph{a priori}, as seen on Fig.~\ref{fig:pokqwdoijjq}.
\end{example}

\begin{example}
For the concurrent system from Example~\ref{exm:3}, whose graph of states is depicted in Fig.~\ref{fig:pkazpoaaaz}, the above technique seems  heavy to derive the probabilistic parameters of the uniform measure. Instead, we rely on the special form $f_\alpha(x)=r^{|x|}\Gamma(\alpha,\alpha\cdot x)$ for the uniform measure, with $r$ the characteristic root of the system, unknown for now, and with $\Gamma$ the Parry cocycle, also unknown.

Put $\lambda=\Gamma(0000,1100)$. For symmetry reasons, it is clear that $\lambda=\Gamma(0000,X)$ for every state $X$ in the middle column of the graph of states depicted on Fig.~\ref{fig:pkazpoaaaz}. Also for symmetry reasons, one also has $\lambda=\Gamma(1111,X)$ for every state in the middle column. All other values of the Parry cocycle can be determined using the cocycle identity, since in particular $\Gamma(\alpha,\alpha)=1$ for every state~$\alpha$. For instance $\Gamma(1100,0000)=\Gamma(1100,1111)=\lambda^{-1}$ and $\Gamma(0000,1111)=1$.

\medskip
Taking into account the symmetry of the system, the Möbius identities $h_\alpha(\vd)=0$ at states $\alpha=0000$ and $\alpha=1100$ write as follows:
\begin{align*}
  1-4 f_{0000}(a)+2f_{0000}(a\cdot c)&=0&
                                    1-2f_{1100}(a)+f_{1100}(a\cdot c)&=0
\end{align*}
Using the form of the valuation $f_\alpha(x)=r^{|x|}\Gamma(\alpha,\alpha\cdot x)$ and using the unknown parameter~$\lambda$, we obtain:
\begin{align*}
  1-4r\lambda+2r^2&=0&1-2\frac r\lambda+r^2&=0
\end{align*}

Putting $u=r\lambda$ and $v=\frac r\lambda$, and after some computations, we obtain the following equations in $u$ and~$v$: $u=v-\frac14$ and $4v^2-9v+4=0$, whence $v=\frac{9\pm\sqrt {17}}8$ and $u=\frac{7\pm\sqrt{17}}8$. The value $v=\frac{9+\sqrt{17}}8$ is seen to lead to a value $r>1$, which is impossible, hence:
\begin{gather}
\left\{  \begin{aligned}
    u&=\frac{7-\sqrt{17}}8\\
    v&=\frac{9-\sqrt{17}}8
  \end{aligned}\right.\qquad\text{yielding\quad}\left\{
  \begin{aligned}
    r&=\frac12\sqrt{5-\sqrt{17}}\approx0.468\\
    \lambda&=\frac{\sqrt{23-\sqrt{17}}}{4\sqrt2}\approx0.768
  \end{aligned}\right.
\end{gather}

Putting $\alpha_0=0000$, and to compute say $h_{\alpha_0}(a)$, we write $h_{\alpha_0}(a)=f_{\alpha_0}(a)-f_{\alpha_0}(a\cdot c)=r\lambda-r^2\approx 0.140$. Other computations are done in a similar way. We obtain thus the following approximate values for the probability law of the first clique, when the initial state of the system is~$0000$:
\begin{gather*}
  \begin{array}{cccccc}
    a&b&c&d&a\cdot c&b\cdot d\\
    0.140&0.140&.140&.140&.219&.219
  \end{array}
\end{gather*}
Contrasting with the previous example, this system has no null node.
\end{example}

\subsection{A comparison result}
\label{sec:an-elem-comp}

In this subsection, we state an elementary lemma and its corollary, both belonging to trace theory, and given in a form slightly more general than precisely needed in the sequel.

Consider an alphabet~$\Sigma$ and two independence relations $I$ and $J$ on $\Sigma$ such that $I\subseteq J$, and consider the two trace monoids $\M=\M(\Sigma,I)$ and $\N=\M(\Sigma,J)$. Then the morphism $\pi_J:\Sigma^*\to\M(\Sigma,J)$ satisfies $\pi_J(ab)=\pi_J(ba)$ for all letters $a$ and $b$ such that $(a,b)\in I$. The universal property of $\M(\Sigma,I)$ as a quotient monoid yields the existence of a surjective morphism $\pi_{I,J}:\M(\Sigma,I)\to\M(\Sigma,J)$ such that $\pi_J=\pi_{I,J}\circ\pi_I$.

It seems to have been unnoticed so far that, when restricted to the set of sub-traces of a given trace of~$\M$, or even of~$\Mbar$, then $\pi_{I,J}$ becomes injective. This is the topic of the following lemma.

The lemma generalizes the following elementary fact. Let $\M=\Sigma^*$ be a free monoid and let $u\in\Sigma^*$. Then any prefix word $x\leq u$ is entirely determined by the collection $(n_a)_{a\in\Sigma}$ where $n_a$ is the number of occurrences of the letter $a$ in~$x$. Hence $x$ is entirely determined by its image in the free commutative monoid generated by~$\Sigma$.

\begin{lemma}
  \label{lem:1}
  Let $I\subseteq J$ be two independence relations on an alphabet\/~$\Sigma$, let $\M=\M(\Sigma,I)$ and $\N=\M(\Sigma,J)$, and let $\pi:\M\to\N$ be the natural surjection. Then $\pi$ extends naturally to a surjection on generalized traces, as a mapping still denoted by $\pi:\Mbar\to\Nbar$. Let $\omega\in\Mbar$, and define:\quad
    $\Mbar_{\leq\omega}=\{x\in\Mbar\tqs x\leq\omega\}$. Then the restriction of\/ $\pi$ to $\Mbar_{\leq\omega}$ is injective.
\end{lemma}

\begin{proof}
  The extension of $\pi$ to a mapping $\Mbar\to\Nbar$ follows from the definitions, hence we focus on proving that the restriction of $\pi$ to $\Mbar_{\leq\omega}$ is injective. Let $x\in\Mbar_{\leq\omega}$ and let $y=\pi(x)$. Let $c_1$ be the first clique in the normal form of~$x$, and let $d_1$ be the first clique in the normal form of~$y$. Let also $C_1$ be the first clique in the normal form of~$\omega$. We assume with loss of generality that $x\neq\vd$ since $\pi^{-1} (\{\vd\})=\{\vd\}$.

  We claim that $c_1=d_1\cap C_1$. The inclusion $c_1\subseteq d_1\cap C_1$ is clear since both inclusions $c_1\subseteq d_1$ and $c_1\subseteq C_1$ are obvious. For proving the converse inclusion, seeking a contradiction, we assume that there is a letter $a\in d_1\cap C_1$ such that $a\notin c_1$. Then, since $y=\pi(x)$, the letter $a$ belongs to some higher clique in the normal form of~$x$. But, since $x\leq\omega$, and since $a\in C_1$, that entails that $a\in c_1$, contradicting the assumption $a\notin c_1$. Hence $c_1=d_1\cap C_1$, as claimed.

  Repeating inductively the same reasoning, with $x'=c_1 \backslash x$ and with $y'=\pi(x')=c_1\backslash y$ and $\omega'=c_1\backslash \omega$ in place of $x$ and of $y$ and of $\omega$ respectively\footnote{Recall that, if $c\leq u$ with $c,u\in\M$, we denote by $c\backslash u$ the left cancellation of $u$ by~$c$, which is the unique trace $v\in\M$ such that $c\cdot v=u$.}, we see that all the cliques $(c_i)_{i\geq1}$ of the generalized trace $x$ can be reconstructed from~$y$. This entails that $\pi$ is injective.
\end{proof}

\begin{corollary}
  \label{cor:1}
  Let $\M$ be a trace monoid, and let $\omega\in\BM$ be an infinite trace. For each integer $n\geq0$, consider:
  \begin{align*}
    \M_{\leq\omega}(n)&=\{x\in\M\tqs x\leq\omega\land |x|=n\},& p_n&=\#\M_{\leq\omega}(n).
  \end{align*}
Then there is a polynomial $P\in\bbZ[X]$ such that $p_n\leq P(n)$ for all integers~$n$. Furthermore, the set\/ $\BM_{\leq\omega}=\{\xi\in\BM\tqs \xi\leq\omega\}$ is at most countable. The polynomial $P$ only depends on~$\M$, and not on~$\omega$.
\end{corollary}

\begin{proof}
  Let $\M=\M(\Sigma,I)$ and let $\N$ be the free commutative monoid generated by~$\Sigma$, \ie, $\N=\M(\Sigma,J)$ with $J=(\Sigma\times\Sigma)\setminus\Delta$ and $\Delta=\{(x,x)\; :\ x\in\Sigma\}$.

  For each integer~$n$, let $q_n=\#\N(n)$. Then it is well known that $q_n=P(n)$ for some polynomial $P\in\bbZ[X]$ (a short proof based on the Möbius inversion formula was given in Sect.~\ref{sec:growth-series-mobius}). Since $I\subseteq J$, it follows  from Lemma~\ref{lem:1} that $p(n)\leq q(n)$.

\medskip
  Furthermore, $\Nbar$~itself is at most countable since $\Nbar$ identifies with:
  \begin{gather*}
    \Nbar\sim\bigl\{(x_i)_{i\in\Sigma}\tq x_i\in\bbZ_{\geq0}\cup\{\infty\},\quad\exists i\in\Sigma\quad x_i=\infty
    \bigr\}.
  \end{gather*}
Hence, the fact that $\BM_{\leq\omega}$ is at most countable also follows from Lemma~\ref{lem:1}.
\end{proof}

\begin{remark}
  Of course, the direct argument:
  \begin{gather*}
\BM_{\leq\omega}\subseteq\bigl\{\xi\in\Cstar^{\bbZ_{\geq1}}\tqs \forall i\geq1\quad C_i(\xi)\subseteq C_i(\omega)\bigr\}
\end{gather*}
would not allow to conclude as in Corollary~\ref{cor:1} that $\BM_{\leq\omega}$ is at most countable.
\end{remark}

\section{Deterministic concurrent systems}
\label{sec:determ-conc-syst}

\begin{definition}
  A \emph{deterministic concurrent system (\DCS)} is a concurrent system $\X=(\M,X,\bot)$ such that for every state $\alpha\in X$, the partial order $(\M_\alpha,\leq)$ is a lattice.
\end{definition}

\begin{remark}
\label{rem:1} According to the background on \lub\ and \glb\ on trace monoids recalled in Section~\ref{sec:lower-upper-bounds} on the one hand, and since $\M_\alpha$ is a downward closed subset of~$\M$ on the other hand, we have for any two executions $x,y\in\M_\alpha$:
  \begin{enumerate}
  \item $x$~and $y$ have a \glb\ in~$\M_\alpha$, which coincides with their \glb\ in~$\M$; and
  \item $x$~and $y$ have a \lub\ in $\M_\alpha$ if and only they have a common upper bound in~$\M_\alpha$, in which case their \lub\ in $\M_\alpha$ coincides with their \lub\ in~$\M$.
  \end{enumerate}
Note however that the existence of $x\vee y$ in $\M$ is not enough to insure that $x\vee y\in\M_\alpha$.

  Henceforth, a concurrent system $(\M,X,\bot)$ is a \DCS\ if and only if, for every state~$\alpha$, any two executions $x,y\in\M_\alpha$ have a common upper bound in~$\M_\alpha$.
\end{remark}

The following result says that \DCS\ correspond to ``locally commutative'' concurrent systems.

\begin{proposition}
  \label{prop:1}
  Let $\X=(\M,X,\bot)$ be a concurrent system. Then the following properties are equivalent:
  \begin{itemize}
  \item[(i)] \label{item:15} $\X$ is deterministic.
  \item[(ii)] \label{item:16} For every $\alpha\in X$, the partial order $(\C_\alpha,\leq)$ is a lattice, isomorphic to $(\P(\Sigma_\alpha),\subseteq)$.
  \end{itemize}
\end{proposition}

\begin{proof}
 The implication $\text{(\ref{item:15})}\implies\text{(\ref{item:16})}$ is obvious. The interesting point is the implication $\text{(\ref{item:16})}\implies\text{(\ref{item:15})}$.

  Assume that $(\C_\alpha,\leq)$ is a lattice for every $\alpha\in X$, which is then necessarily isomorphic to $(\P(\Sigma_\alpha),\subseteq)$. Fix $\alpha\in X$ and let $x,y\in\M_\alpha$. Assume first that $x\wedge y=\vd$. Let $(c_1,\ldots,c_k)$ and $(d_1,\ldots,d_m)$ be the normal forms of $x$ and of~$y$. Maybe by adding the empty trace at the tail of one or the other normal form, we assume that $k=m$, at the cost of tolerating that some of the elements may be the empty trace.

  On the one hand, since $c_1\cdot c_2$ is an execution starting from~$\alpha$, one has $c_2\in\C_{\alpha\cdot c_1}$. On the other hand, both $c_1$ and $d_1$ belong to~$\C_\alpha$, which is a lattice by assumption. Hence $c_1\vee d_1\in\C_\alpha$. And since $c_1\wedge d_1=\vd$ by assumption, one has $c_1\vee d_1=c_1\cdot d_1=d_1\cdot c_1$. Therefore: $d_1\in\C_{\alpha\cdot c_1}$. Since both cliques $c_2$ and $d_1$ belong to~$\C_{\alpha\cdot c_1}$, which is a lattice, it follows that $c_2\vee d_1\in\C_{\alpha\cdot c_1}$.

  Now we claim that $c_2\wedge d_1=\vd$. Otherwise, there exists a letter $a$ occurring in both $c_2$ and~$d_1$. Since $(c_1,c_2)$ is a normal pair of cliques, there exists $b\in c_1$ such that $(a,b)\in D$, the dependence pair of the monoid. Because of the assumption $c_1\wedge d_1=\vd$, the identity $a=b$ is impossible. But both $a$ and $b$ belong to~$\Sigma_\alpha$, and since $a\neq b$, the fact that $(a,b)\in D$ contradicts that $\C_\alpha$ is a lattice; our claim is proved.

\medskip
  We have obtained that $c_2\vee d_1$ exists in $\C_{\alpha\cdot c_1}$ and that $c_2\wedge d_1=\vd$. Hence $c_2\vee d_1=c_2\cdot d_1=d_1\cdot c_2$. It implies that $c_2\in\C_{\alpha\cdot(c_1\vee d_1)}$. Symmetrically, we obtain that $d_2\in\C_{\alpha\cdot (c_1\vee d_1)}$. Since $\C_{\alpha\cdot(c_1\vee d_1)}$ is a lattice, it follows that $d_2\vee c_2\in\C_{\alpha\cdot(c_1\vee d_1)}$. But again,  $d_2\wedge c_2=\vd$ hence $d_2\vee c_2=d_2\cdot c_2=c_2\cdot d_2$. Therefore we obtain that the following trace belongs to~$\M_\alpha$:
    \begin{gather*}
      (c_1\vee d_1)\cdot(c_2\vee d_2)=(c_1\cdot c_2)\cdot(d_1\cdot d_2)=(d_1\cdot d_2)\cdot(c_1\cdot c_2).
    \end{gather*}
    Repeating inductively the same reasoning, we finally obtain that $x\cdot y=y\cdot x\in\M_\alpha$, thus providing a common upper bound of $x$ and of $y$ in~$\M_\alpha$. This proves the existence of $x\vee y$ in $\M_\alpha$ in the case where $x\wedge y=\vd$.

\medskip
    The general case follows by considering $x'=(x\wedge y)\backslash x$ and $y'=(x\wedge y)\backslash y$ instead of $x$ and~$y$.
\end{proof}

\begin{remark}
In a \DCS, for each state $\alpha\in X$, the partially ordered set of cliques $(\C_\alpha,\leq)$ identifies with the powerset $(\P(\Sigma_\alpha),\subseteq)$. In particular $\C_\alpha$ has a maximum $c_\alpha=\max(\C_\alpha)=\bigvee\Sigma_\alpha$, given by:\quad $c_\alpha=\Sigma_\alpha$. We keep this notation in the statement of the following lemma.
\end{remark}

\begin{lemma}
  \label{lem:3}
  Let $\X=(\M,X,\bot)$ be a deterministic concurrent system, and let  $\alpha\in X$. Let $T_\alpha=(c_{i})_{i\geq1}$ be the sequence of cliques defined by $c_{1}=c_\alpha$, and inductively by $c_{i+1}=c_{\alpha_{i}}$ where $\alpha_i=\alpha\cdot (c_{1}\cdot\ldots\cdot c_{i})$. Then $T_\alpha$ is a generalized execution which is the maximum of\/ $(\Mbar_\alpha,\leq)$.
\end{lemma}

\begin{proof}
  We first observe that, for $c_\alpha$ the maximum of~$\C_\alpha$, then $c_\alpha\to y$ holds\footnote{This actually holds for any concurrent system, not necessarily deterministic, if $c_\alpha$ is taken to be any maximal element in~$\C_\alpha$.} for every clique $y\in\C_{\alpha\cdot c_\alpha}$. Here in particular, $c_i\to c_{i+1}$ holds for all $i\geq1$, hence $T_\alpha$ is indeed a generalized execution.

  Let $x\in\Mbar_\alpha$, with $x=(d_i)_{i\geq1}$. We prove that $x\leq T_\alpha$. Assume first that $x$ is a finite trace, of height~$k$. Put $y=c_1\cdot\ldots\cdot c_k$. Then $x$ and $y$ belong to~$\M_\alpha$. Hence $z=x\vee y$ exists in~$\M_\alpha$. Let $(e_1,\ldots,e_k)$ be the normal form of~$z$ (since $x$ and $y$ have the same height~$k$, $z$ also has height~$k$). Then $c_j\leq e_j$ and thus $c_j=e_j$ for all~$j$ by maximality of~$c_j$. Hence $d_j\leq c_j$ for all~$j$, which was to be proved.

  If $x=(c_i)_{i\geq1}$ is now a generalized trace, we obtain the same result by applying the previous case to all the sub-traces $(c_i)_{1\leq i\leq k}$, for $k$ ranging over the positive integers.
\end{proof}

Let us introduce a name for a valuation that will play a special role.

\begin{definition}
  \label{def:2}
  Let $\X=(\M,X,\bot)$ be a concurrent system. The valuation $f=(f_\alpha)_{\alpha\in X}$ defined by:
  \begin{gather*}
    \forall\alpha\in X\quad \forall x\in\M\quad f_\alpha(x)=
    \begin{cases}
      1,&\text{if\/ $x\in\M_\alpha$}\\
      0,&\text{otherwise}
    \end{cases}
  \end{gather*}
is called the \emph{dominant valuation} of~$\X$.
\end{definition}

The family $f=(f_\alpha)_{\alpha\in X}$ given in Def.~\ref{def:2} is indeed a valuation. Indeed, using the axioms of the monoid action and the additional assumption $\bot\cdot z=\bot$ for all $z\in\M$, one sees that the following equivalence is true for every $\alpha\in X$ and for every traces $x,y\in\M$:
\begin{gather*}
  \alpha\cdot (x\cdot y)\neq\bot\iff(\alpha\cdot x\neq\bot\land (\alpha\cdot x)\cdot y\neq\bot),
\end{gather*}
which translates at once as the identity $f_\alpha(x\cdot y)=f_\alpha(x)f_{\alpha\cdot x}(y)$.

\begin{theorem}
  \label{thr:1}
  Let $\X=(\M,X,\bot)$ be a non trivial concurrent system.
  \begin{enumerate}
  \item\label{item:12} If\/ $\Sigma_\alpha\neq\emptyset$ for all $\alpha\in X$, then the two following statements are equivalent:
  \begin{itemize}
  \item[(i)] \label{item:1} $\X$ is deterministic.
  \item[(ii)] \label{item:2} The dominant valuation of $\X$ is probabilistic.
  \end{itemize}
\item\label{item:13} If $\X$ is deterministic, then all sets~$\BM_\alpha$, for $\alpha\in X$,  are at most countable and the characteristic root of $\X$ is $r=1$ or $r=\infty$.
  \end{enumerate}
\end{theorem}

\begin{proof}
Point~\ref{item:12}. To prove the stated equivalence, assume~(i), and let $f=(f_\alpha)_{\alpha\in X}$ be the dominant valuation. Let $\alpha\in X$, and let $c\in\C_\alpha$. Since $\C_\alpha$ identifies with~$\P(\Sigma_\alpha)$, the Möbius transform of $f_\alpha$ evaluated at $c$ is given by:
  \begin{align*}
    h_\alpha(c)&=\sum_{c'\in\C_\alpha\tqs c'\geq c}(-1)^{|c'|-|c|}=
                 \begin{cases}
                   1,&\text{if $c=c_\alpha$ (the maximum of $\C_\alpha$)}\\
                   0,&\text{otherwise}.
                 \end{cases}
  \end{align*}
  Since $\vd\neq c_\alpha$ for all $\alpha\in X$, this shows that $f$ is a probabilistic valuation.

\medskip
  Conversely, assume as in~(ii) that $f$ is probabilistic. Let $\alpha\in X$ be a state, and let $c_\alpha$ be a maximal element of~$(\C_\alpha,\leq)$. Then, on the one hand, and since $c_\alpha$ is a maximal clique, one has $h_\alpha(c_\alpha)=f_\alpha(c_\alpha)=1$. But on the other hand, $h_\alpha$~is nonnegative on $\C_\alpha$ and sums up to~$1$ on~$\C_\alpha$. Hence $h_\alpha$ vanishes on all other cliques of~$\C_\alpha$. Since this is true for every maximal element of~$\C_\alpha$, it entails that $\C_\alpha$ has actually a unique maximal element, which is thus its maximum~$\Sigma_\alpha$. Hence $(\C_\alpha,\leq)$ is a lattice for every $\alpha\in X$, which proves~(i) according to Proposition~\ref{prop:1}.

  Point~\ref{item:13}. We assume that $\X$ is a \DCS. According to Lemma~\ref{lem:3}, the partial order $(\Mbar_\alpha,\leq)$ has a maximum~$T_\alpha$ for every $\alpha\in X$, hence $\Mbar_\alpha\subseteq\Mbar_{\leq T_\alpha}$.  It follows at once from Corollary~\ref{cor:1} that $\BM_\alpha$ is at most countable, and that $\#\M_\alpha(n)\leq P(n)$ for all integers~$n$ and for some polynomial~$P$. All generating series $G_{\alpha,\beta}(z)$ are rational with non zero coefficients at least~$1$, and they have their coefficients dominated by some polynomial. They have therefore a radius of convergence either $1$ or~$\infty$. Hence $r\in\{1,\infty\}$.
\end{proof}

\begin{remark}
  In general, there might exist other probabilistic valuations than the dominant valuation, even for a \DCS. See Example~\ref{exm:5} at the end of next section.
\end{remark}

Since the dominant valuation $f$ is probabilistic, it corresponds to a Markov measure as described in Sect.~\ref{sec:valu-prob-valu}. The behavior of the resulting Markov chain of states-and-cliques is trivial, as shown by the following result.

\begin{proposition}
  \label{prop:2}
  Let $\X=(\M,X,\bot)$ be a non trivial \DCS\ such that\/ $\Sigma_\alpha\neq\emptyset$ for all $\alpha\in X$, and let $\nu=(\nu_\alpha)_{\alpha\in X}$ be the Markov measure associated with the dominant valuation. Then for each initial state $\alpha\in X$, the probability measure $\nu_\alpha$ is the Dirac distribution~$\delta_{\{T_\alpha\}}$, where $T_\alpha=\max\Mbar_\alpha$.

Furthermore, with respect to the dominant valuation, every node of the digraph of states-and-cliques is null except for those of the form~$(\alpha,c_\alpha)$, with $c_\alpha=\bigvee\C_\alpha=\Sigma_\alpha$.
\end{proposition}

\begin{proof}
  Assuming that $\X$ is a \DCS, we keep using the notation $c_\alpha=\max\C_\alpha=\Sigma_\alpha$ for all $\alpha\in X$.

\medskip
  A direct proof is as follows. Fix $\alpha\in X$, and let $(\alpha_{i},z_i)_{i\geq0}$ be defined inductively by $\alpha_0=\alpha$, $z_0=\vd$ and $z_{i+1}=z_i\cdot c_{\alpha_i}$, $\alpha_{i+1}=\alpha\cdot z_i$. On the one hand, we have $\bigvee_{i\geq0}z_i=T_\alpha$ by the construction used in the proof of Lemma~\ref{lem:3}. But on the other hand, the characterization of the probability measure $\nu_\alpha$ yields $\nu_\alpha(\up z_i)=f(z_i)=1$ for all $i\geq0$. Since $\up z_{i+1}\subseteq\up z_{i}$ for all $i\geq0$, we have thus:
  \begin{gather*}
    \nu_\alpha(\omega\geq T_\alpha)=\nu_\alpha\Bigl(\,\bigcap_{i\geq0}\up z_i\Bigr)=\lim_{i\to\infty}\nu_\alpha(\up z_i)=1.
  \end{gather*}
  Since $T_\alpha=\max\Mbar_\alpha$, it implies $\nu_\alpha(\omega=T_\alpha)=1$.

\medskip
  An alternative proof is as follows. Let $(Y_i)_{i\geq1}$ be the Markov chain of states-and-cliques associated to the dominant valuation, and let $\alpha\in X$. One has $\nu_\alpha(C_1=c)=h_\alpha(c)$ for all $c\in\Cstar_\alpha$, by~(\ref{eq:9}). The values of $h_\alpha$ computed in the proof of Th.~\ref{thr:1} show that the initial distribution of the chain is~$\delta_{\{(\alpha,c_\alpha)\}}$. It is shown in \cite{abbes19:_markov} that the $(\alpha,c)$-row of the transition matrix of the chain is proportional to~$h_{\alpha\cdot c}(\cdot)$. Hence all entries of the  $(\alpha,c)$-row are~$0$, except for the $\bigl((\alpha,c),(\beta,c_\beta)\bigr)$ entry with $\beta=\alpha\cdot c$, where the entry is~$1$. Hence the execution $T_\alpha$ is given $\nu_\alpha$-probability~$1$.

  Finally we prove the statement about null nodes. The formula~\eqref{eq:1} defining the Möbius transform shows that $h_\alpha(c_\alpha)=1$ (for this we use the fact that $\Sigma_\alpha\neq\emptyset$). Since $\bigl(h_\alpha(c)\bigr)_{c\in\Cstar_\alpha}$ is a probability vector, it entails that all other cliques $c\in\Cstar_\alpha$ satisfy $h_\alpha(c)=0$, hence $(\alpha,c)$ is a null node if $c\neq c_\alpha$.
\end{proof}

\section{Irreducible deterministic concurrent systems}
\label{sec:irred-determ-conc}

Before stating the main result of this section, we prove two lemmas.

\begin{lemma}
  \label{lem:2}
  Let $\X=(\M,X,\bot)$ be a \DCS. Let $\alpha\in X$ and let $c\in\C_\alpha$ be a clique such that $a\notin c$ for some letter $a\in\Sigma_\alpha$. Then:
  \begin{gather*}
    \forall x\in\Mbar_\alpha\quad C_1(x)=c\implies a\notin x.
  \end{gather*}
\end{lemma}

\begin{proof}
  Let $\alpha$, $a$ and $c$ be as in the statement. Clearly, the implication stated in the lemma is true if we prove it to be true for $x$ ranging over~$\M_\alpha$ instead of~$\Mbar_\alpha$. Hence, let $x\in\M_\alpha$ be such that $C_1(x)=c$. Let $(c_i)_{i\geq1}$ be the generalized normal form of~$x$, and define by induction $x_0=\vd$, $x_{i+1}=x_i\cdot c_{i+1}$ for all $i\geq0$ and $\alpha_i=\alpha\cdot x_i$ for all $i\geq0$. We prove by induction on $i\geq1$ that:
  \begin{enumerate}
  \item $a\in\Sigma_{\alpha_{i-1}}$; and
  \item $a\notin c_i$.
  \end{enumerate}

  For $i=1$, both properties derive from the assumptions of the lemma. Assume that both properties hold for some $i\geq1$. By construction, $c_i\in\C_{\alpha_{i-1}}$\,, and $a\in\Sigma_{\alpha_{i-1}}$ by the induction hypothesis. Since the concurrent system is deterministic, it follows that $a\vee c_i\in\C_{\alpha_{i-1}}$. Since $a\notin c_i$ by the assumption hypothesis, this \lub\ is given by $c_i\cdot a\in\C_{\alpha_{i-1}}$\,. This entails first that $a\in\C_{\alpha_{i-1}\cdot c_i}$\,, but $\alpha_{i-1}\cdot c_i=\alpha_i$ hence $a\in\Sigma_{\alpha_i}$\,. But it also entails that $a\notin c_{i+1}$\,, completing the induction step. The result of the lemma follows.
\end{proof}

\begin{lemma}
  \label{lem:4}
  Let $\X=(\M,X,\bot)$ be a concurrent system. Let $\alpha\in X$, and let $r_\alpha$ be the radius of convergence of the generating series $G_\alpha(z)=\sum_{x\in\M_\alpha}z^{|x|}$. Then the following properties are equivalent:
  \begin{itemize}
  \item[(i)] \label{item:5} $\M_\alpha$~is finite;
  \item[(ii)] \label{item:6} $\BM_\alpha=\emptyset$;
  \item[(iii)] \label{item:7} $r_\alpha=\infty$.
  \end{itemize}
\end{lemma}

\begin{proof}
  The implications $\text{(\ref{item:5}})\implies\text{(\ref{item:6}})$ and $\text{(\ref{item:5}})\implies\text{(\ref{item:7}})$ are clear.

\medskip
Assume that $\M_\alpha$ is infinite. Then there exist executions in $\M_\alpha$ of length arbitrary large. Therefore there exist $x\in\M_\alpha$ and $y\neq\vd$ such that $\alpha\cdot x=\alpha\cdot(x\cdot y)$. Then all traces $x_n=x\cdot y^n$ belong to $\M_\alpha$ for $n\geq0$. This proves two things. First, if $k=|y|$, the  coefficient of $z^{|x|+kn}$ in the series $G_\alpha(z)$ is~$\geq1$ for all integers~$n$, hence $r_\alpha<\infty$. Second, the execution $\xi=\bigvee_{n\geq0}x_n$ is an element of~$\BM_\alpha$, showing that $\BM_\alpha\neq\emptyset$. Hence we have proved both  $\text{(\ref{item:6}})\implies\text{(\ref{item:5}})$ and $\text{(\ref{item:7}})\implies\text{(\ref{item:5}})$ by contraposition, completing the proof.
\end{proof}

\begin{theorem}
  \label{thr:2}
  Let $\X=(\M,X,\bot)$ be an irreducible concurrent system, of characteristic root~$r$, and let $f$ be the dominant valuation of~$\X$. Then the following statements are equivalent:
    \begin{itemize}
    \itemsep=0.85pt
  \item[(i)] \label{item:3} $\X$ is deterministic.
  \eject
  \item[(ii)] \label{item:4} $f$ is a probabilistic valuation.
  \item[(iii)] \label{item:8} $f$ is the only probabilistic valuation of~$\X$.
  \item[(iv)] \label{item:9} $r=1$.
  \item[(v)] \label{item:10} One set\/ $\BM_\alpha$ is at most countable.
  \item[(vi)] \label{item:11} Every set\/ $\BM_\alpha$ is at most countable.
  \item[(vii)] \label{item:14} One set $\BM_\alpha$ is a singleton.
  \item[(viii)] \label{item:15a} Every set $\BM_\alpha$ is a singleton.
  \end{itemize}
\end{theorem}

\begin{proof}
Since $\X$ is irreducible, it satisfies in particular  $\Sigma_\alpha\neq\emptyset$ for all $\alpha\in X$. Hence the equivalence $\text{(\ref{item:3})}  \iff\text{(\ref{item:4})}$ and the implications $\text{(\ref{item:3})}\implies\text{(\ref{item:9})}$ and $\text{(\ref{item:3})}\implies\text{(\ref{item:11})}$ derive already from Theorem~\ref{thr:1}. The implications $\text{(\ref{item:8})}\implies\text{(\ref{item:4})}$, $\text{(\ref{item:11})}\implies\text{(\ref{item:10})}$ and $\text{(\ref{item:15a})}\implies\text{(\ref{item:14})}\implies\text{(\ref{item:10})}$ are trivial.

\smallskip
  $\text{(\ref{item:3})}\implies\text{(\ref{item:8})}$. Let $f=(f_\alpha)_{\alpha\in X}$ be a probabilistic valuation, and let $\widetilde f=(\widetilde f_\alpha)_{\alpha\in X}$ be the dominant valuation. Let $\alpha\in X$ and let $c\in\Cstar_\alpha$ with $c\neq c_\alpha$, where $c_\alpha=\Sigma_\alpha$ is the maximum of~$\C_\alpha$. There is thus a letter $a\in\Sigma_\alpha$ such that $a\notin c$. Let $\M^a$ be the submonoid of $\M$ generated by $\Sigma\setminus\{a\}$. It follows from Lemma~\ref{lem:2} that $\{\omega\in\BM_\alpha\tqs C_1(\omega)=c\}\subseteq\BM^a_\alpha$.

  According to the spectral property recalled in Section~\ref{sec:irreducibility}, the characteristic root $r^a$ of $\X^a=(\M^a,X,\bot)$ satisfies $r^a>r$ since $\X$ is assumed to be irreducible. But $r=1$ since $\X$ is deterministic, and therefore $r^a=\infty$, which implies that $\BM^a_\alpha=\emptyset$ according to Lemma~\ref{lem:4}. Let $\nu=(\nu_\alpha)_{\alpha\in X}$ be the family of probability measures associated with the probabilistic valuation~$f$, as explained in Sect.~\ref{sec:valu-prob-valu}. Then $\nu_\alpha(\BM_\alpha^a)=0$ and thus $\nu_\alpha(C_1=c)=0$. But one also has $h_\alpha(c)=\nu_\alpha(C_1=c)$ according to~(\ref{eq:9}), where $h_\alpha$ is the Möbius transform of~$f_\alpha$. Hence $h_\alpha(c)=0$. We have proved that $h_\alpha$ vanishes on all cliques $c\in\C_\alpha$ such that $c\neq c_\alpha$. Since $(h_\alpha(c))_{c\in\Cstar_\alpha}$ is a probability vector, it entails that $h_\alpha(c_\alpha)=1$. Thus $h_\alpha$ coincides with the Möbius transform of~$\widetilde f_\alpha$, and $f=\widetilde f$.

\smallskip
  $\text{(\ref{item:9})}\implies\text{(\ref{item:3})}$ and $\text{(\ref{item:10})}\implies\text{(\ref{item:3})}$. By contraposition, assume that $\X$ is not deterministic. Prop.~\ref{prop:1} implies the existence of a state $\alpha$ and of two distinct letters $a,b\in\Sigma_\alpha$ such that $a\cdot b\neq b\cdot a$. Since $\X$ is assumed to be irreducible, there exist $x\in\M_{\alpha\cdot a,\alpha}$ and $y\in\M_{\alpha\cdot b,\alpha}$. Put $x_a=a\cdot x$ and $x_b=b\cdot y$, and we can also assume without loss of generality that $|x_a|=|x_b|$. Then $\M_\alpha$ contains the submonoid generated by $\{x_a,x_b\}$, which is free. This implies two things: first,  the generating series $G_\alpha(z)=\sum_{x\in\M_\alpha}z^{|x|}$ has radius of convergence smaller than~$1$, and thus $r<1$; second, $\BM_\alpha$~is uncountable.

\smallskip
  $\text{(\ref{item:3})}\implies\text{(\ref{item:15a})}$. Seeking a contradiction, assume that $\X$ is irreducible and that for some state $\alpha\in X$, the set $\BM_\alpha$ has two distinct elements. Since $T_\alpha=\max\BM_\alpha$ is already an element of~$\BM_\alpha$, there is thus $\omega\in \BM_\alpha$ with $\omega\neq T_\alpha$. Let $\omega=(c_i)_{i\geq1}$ and let $c_1,\ldots,c_k$ be the longest initial sequence of cliques that the two infinite traces $\omega$ and $T_\alpha$ have in common. Put $x=c_1\cdot\ldots\cdot c_k$ and $\beta=\alpha\cdot x$, and $\xi=(c_i)_{i>k}$. Then $\xi\in\BM_\beta$, and by construction there is a letter $a\in\Sigma_\beta$ such that $a\notin c_{k+1}$. It follows from Lemma~\ref{lem:2} that $a\notin c_i$ for all $i>k$.

  Consider the restricted concurrent system $\X'=\X^a$. Then, since $\X$ is irreducible, it follows from the spectral property that $r^a>1$, hence $r^a=\infty$. According to Lemma~\ref{lem:4}, it implies that $\BM^a_\alpha=\emptyset$ for every $\alpha\in X$. Yet, the infinite trace $\xi\in\BM_\beta$ found earlier satisfies $\xi\in\BM^a_\beta$, a contradiction.

  The proof is complete.
\end{proof}

For an irreducible \DCS\ equipped with its dominant valuation, Prop.~\ref{prop:2} applies. Hence the Markov chain of states-and-cliques associated to this unique probabilistic valuation follows the trivial dynamics described by Prop.~\ref{prop:2}. In particular, the null nodes are easy to detect: all nodes of $\Dstar$ of the form $(\alpha,c)$ with $\alpha\neq c_\alpha$. Consequently:

\begin{corollary}
  \label{cor:2}
  The Markov chain of states-and-cliques of an irreducible \DCS\ stays within a subgraph of\/ $\Dstar$ of size~$\# X$. And there is no cycle in $\Dstar$ connecting null nodes.
\end{corollary}

\begin{proof}
  According to Prop.~\ref{prop:2}, to each state $\alpha\in X$ is associated a unique non null node, namely $(\alpha,c_\alpha)$. The Markov chain of states-and-cliques visits only non null nodes, whence the result.

  To prove the second statement, assume the existence of a cycle in $\Dstar$ connecting null nodes only. It yields the existence of an infinite executing, starting from some state~$\alpha$, which stays within null nodes in~$\Dstar$. According to Th.~\ref{thr:2}, there is a unique infinite execution starting from~$\alpha$, which is~$T_\alpha$. Hence $T_\alpha$ stays within null nodes of~$\Dstar$, a contradiction.
\end{proof}

\begin{remark}
The first statement in the above corollary does not mean that the graph of non null nodes in $\Dstar$ is itself strongly connected. The next example illustrates this fact.
\end{remark}

\begin{remark}
In general, if a concurrent system is not deterministic, null nodes of $\Dstar$ may be connected by a cycle. This is the case for the Petri net from  Example~\ref{exm:4}, whose digraph of states-and-cliques is depicted in Fig.~\ref{fig:pokqwdoijjq}, with the cycle $\bigl((\alpha_0,d),(\alpha_0,d)\bigr)$. It may also happen for a deterministic concurrent system if it is not irreducible. The reader may check it for the free commutative monoid on two generators seen as a \DCS.
\end{remark}

The following example illustrates the dynamics of an irreducible \DCS.

\begin{example}
Figure~7  depicts an example of irreducible \DCS. The digraph of states-and-cliques of the system is depicted on Fig.~\ref{fig:aqqqapokferw}. Compare with the situation depicted on next example for a \DCS\ which is not irreducible.

  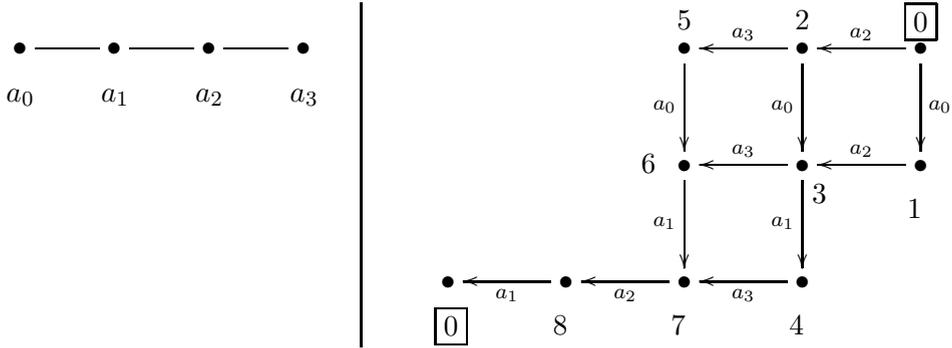
\begin{figure}[!h]
  \vspace{2mm}
  \begin{center}
 $$
    \begin{array}{c|c}
   \xymatrix{%
      \bullet\ar@{-}[r]\labeld{a_0}
         &\bullet\ar@{-}[r]\labeld{a_1}
         &\bullet\ar@{-}[r]\labeld{a_2}
         &\bullet\labeld{a_3}
  }
 \quad\strut      &\qquad
      \xymatrix@R=3em@C=3em{
      &&\bullet\ar[d]_{a_0}\labelu{5}
      &\bullet\ar[d]_{a_0}\ar[l]_{a_3}\labelu{2}
      &\bullet\ar[d]^{a_0}\ar[l]_{a_2}\labelu{\fbox{$0$}}\\
      &&\bullet\ar[d]_{a_1}\labell{6}
      &\bullet\ar[d]_{a_1}\ar[l]_{a_3}\labeldr{3}
      &\bullet\ar[l]_{a_2}\labeld{1}\\
      \bullet\labeld{\fbox{$0$}}
      &\bullet\ar[l]^{a_1}\labeld{8}
      &\bullet\ar[l]^{a_2}\labeld{7}
      &\bullet\ar[l]^{a_3}\labeld{4}
  }
    \end{array}
 $$
 \end{center}
  \caption{Example of an irreducible and deterministic concurrent system $\X =
(\M, X, \bot)$ with $\Sigma = \{a_0, \ldots, a_3\}$, $X = \{0, 1,\ldots, 8\}$.
{\em Left:} Coxeter graph of the monoid {$\M$}.
{\em Right:} multigraph of states of {$\X$}. The two framed labels {\fbox{$0$}} are identified and correspond to the same state} \label{fig:aqqqwqw}
  \end{figure}

  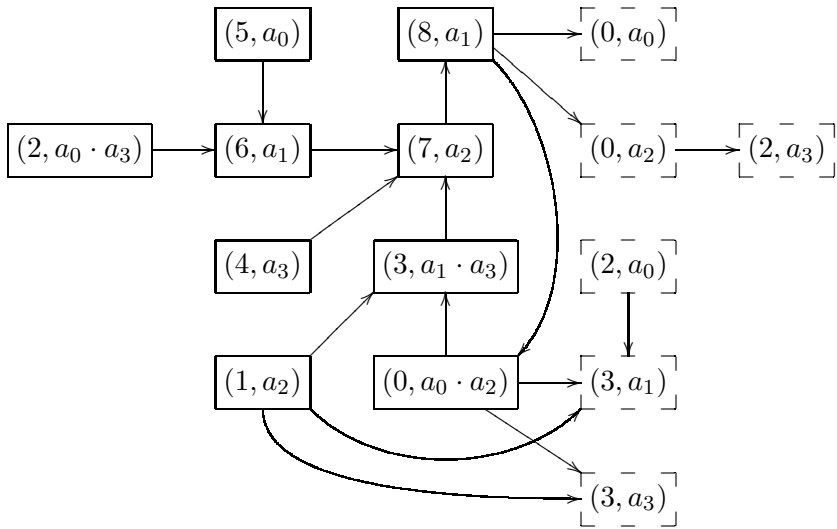
\begin{figure}[!h]
  \vspace{2mm}
    \begin{center}
    $$
    \xymatrix{%
&*+[F]{\strut(5,a_0)}\ar[d]      &*+[F]{\strut(8,a_1)}\ar[r]\POS!D!R\ar@(dr,ur)[ddd]!R!U\POS[]\POS!R!D(.5)\ar[dr]!U!L
      &*+[F--]{\strut(0,a_0)}\\
*+[F]{\strut(2,a_0\cdot a_3)}\ar[r]&      *+[F]{\strut(6,a_1)}\ar[r]
      &*+[F]{\strut(7,a_2)}\ar[u]
      &*+[F--]{\strut(0,a_2)}\ar[r]
      &*+[F--]{\strut(2,a_3)}
      \\
&*+[F]{\strut(4,a_3)}\POS!U!R\ar[ur]!D!L      &*+[F]{\strut(3,a_1\cdot a_3)}\ar[u]
      &*+[F--]{\strut(2,a_0)}\ar[d]
\\
       &*+[F]{\strut(1,a_2)}\POS!U!R\ar[ur]!D!L\POS!D!R\ar@(dr,dl)[rr]!D!L\POS[]\POS!D\ar@(d,l)[drr]
     &*+[F]{\strut(0,a_0\cdot a_2)}\ar[u]\ar[r]\ar[dr]!U!L
      &*+[F--]{\strut(3,a_1)}
\\
&&&*+[F--]{\strut(3,a_3)}
      }
    $$
    \end{center}
    \caption{\small Digraph of states-and-cliques for the \DCS\ depicted on Fig.~\ref{fig:aqqqwqw}. Nodes with solid frames are nodes of the form $(\alpha,c_\alpha)$ with $c_\alpha=\max\C_\alpha$. Nodes with a dashed frame are null nodes. The probability for the Markov chain of states-and-cliques to jump from a solid frame node to a dashed frame node is~$0$; the probability of starting in a dashed node is~$0$.} \label{fig:aqqqapokferw}
    \end{figure}\vspace*{-2mm}
\end{example}

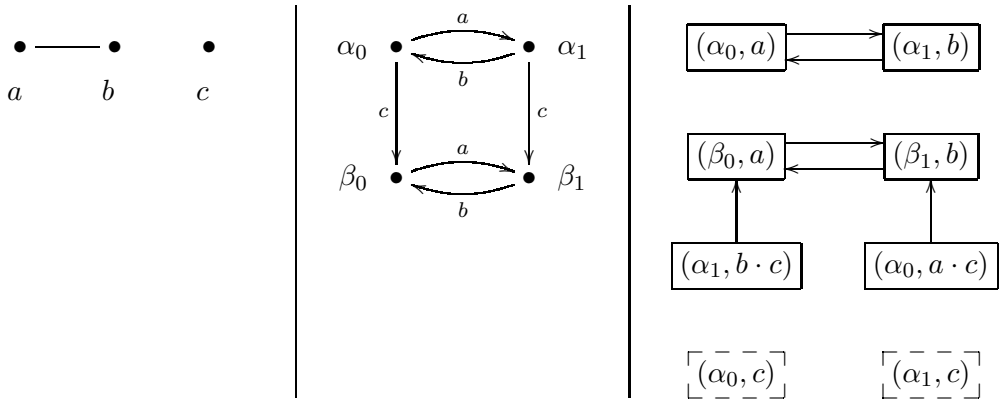
\begin{figure}[!ht]
  $$\begin{array}{c|c|c}
      \xymatrix{\bullet\ar@{-}[r]\labeld{a}&\bullet\labeld{b}&\bullet\labeld{c}}
\qquad\strut&\qquad
\xymatrix@R=3.5em@C=3.5em{\bullet\ar@/^/^{a}[r]\ar[d]_ {c}\labell{\alpha_0}
&\bullet\ar@/^/^{b}[l]\ar[d]^{c}\labelr{\alpha_1}\\
\bullet\ar@/^/^{a}[r]\labell{\beta_0}&\bullet\ar@/^/^{b}[l]\labelr{\beta_1}
                                       }\qquad\strut
&\quad
\entrymodifiers={+[][F]}%
\xymatrix{
(\alpha_0,a)\ar@<1ex>[r]&(\alpha_1,b)\ar@<1ex>[l]\\
      (\beta_0,a)\ar@<1ex>[r]&(\beta_1,b)\ar@<1ex>[l]\\
      (\alpha_1,b\cdot c)\ar[u]&(\alpha_0,a\cdot c)\ar[u]\\
      *+[F--]{(\alpha_0,c)}&*+[F--]{(\alpha_1,c)}
}
  \end{array}$$
  \caption{\small A non irreducible \DCS\ not satisfying property (\ref{item:8}) of Th.~\ref{thr:2}. \textsl{Left:} the Coxeter graph of the monoid. \textsl{Middle:} the multigraph of states of the \DCS. \textsl{Right:} the digraph of states-and-cliques. The parameter $p$ is only involved in the initial distribution of the Markov chain of states-and-cliques. The dashed nodes are the null nodes and they are immaterial to the Markov chain of states-and-cliques}
  \label{fig:;lksadfk}
\end{figure}

Without the irreducibility assumption, the equivalence stated in Th.~\ref{thr:2} may fail. We give below an example of a deterministic concurrent systems not irreducible, and not satisfying point~(\ref{item:8}).

\begin{example}
  \label{exm:5}
Let $\X=(\M,X,\bot)$ be the \DCS\ depicted in Fig.~\ref{fig:;lksadfk}. The system is not irreducible for several reasons: none of the three conditions for irreducibility is met. The probabilistic valuations of $\X$ are all of the following form, for some real $p\in[0,1]$:
  \begin{align*}
    f_{\alpha_0}(a)&=1&f_{\alpha_0}(c)&=p&    f_{\alpha_1}(b)&=1&f_{\alpha_1}(c)&=p&f_{\beta_0}(a)&=1&f_{\beta_1}(b)&=1
  \end{align*}

  Hence the dominant valuation is not the unique probabilistic valuation, contrary to irreducible systems as stated by point (\ref{item:8}) of Th.~\ref{thr:2}. The parameter $p$ is to be interpreted as the ``probability of playing~$c$'' in the course of the execution.  But this decision---playing $c$ or not---is made once, hence allowing all values between $0$ or~$1$ for the probability. Whereas, in a sequential model of concurrency, that would typically be a decision repeated infinitely often, hence yielding the only two possible values $0$ or $1$ for this probability.
The formula $\nu_\alpha(C_1=\gamma)=h_\alpha(\gamma)$ for $\gamma\in\Cstar_\alpha$ yields the following initial distribution of the Markov chain of states-and-cliques if, for instance, the initial state of the system is~$\alpha_0$:
$ \nu_{\alpha_0}(C_1=a)=1-p$,\quad $\nu_{\alpha_0}(C_1=c)=0$,\quad$\nu_{\alpha_0}(C_1=ac)=p$.
\end{example}

\end{document}